\RequirePackage{fix-cm}
\documentclass[twocolumn]{svjour3}          % twocolumn
\smartqed  % flush right qed marks, e.g. at end of proof

\usepackage{amssymb}
\usepackage{graphicx}
\usepackage{blkarray}
\usepackage{amsmath}
\usepackage{subfig}
\usepackage{balance}
\journalname{Social Network Analysis and Mining}
\begin{document}

\title{HellRank: A Hellinger-based Centrality Measure \\for Bipartite Social Networks}


\author{Seyed Mohammad Taheri, \\
Hamidreza Mahyar, \\ Mohammad Firouzi, \\ Elahe Ghalebi K., \\ Radu Grosu, \\ Ali Movaghar %etc.
}


\institute{S.M. Taheri, M. Firouzi, A. Movaghar \at
              Department of Computer Engineering, Sharif University of Technology (SUT) \\
              \email{\{mtaheri, mfirouzi\}@ce.sharif.edu, \\movaghar@sharif.edu}           %  \\
           \and
           H. Mahyar, E. Ghalebi K., R. Grosu \at
              Department of Computer Engineering, Vienna University of Technology (TU Wien) \\
              \email{\{hmahyar, eghalebi\}@cps.tuwien.ac.at, \\radu.grosu@tuwien.ac.at}
}

\date{Received: 2016-10 / Accepted: date}

\maketitle

\begin{abstract}
Measuring centrality in a social network, especially in bipartite mode, poses many challenges. For example, the requirement of full knowledge of the network topology, and the lack of properly detecting top-$k$ \textit{behavioral representative users}. To overcome the above mentioned challenges, we propose HellRank, an accurate centrality measure for identifying central nodes in bipartite social networks. HellRank is based on the Hellinger distance between two nodes on the same side of a bipartite network. We theoretically analyze the impact of this distance on a bipartite network and find upper and lower bounds for it. The computation of the HellRank centrality measure can be distributed, by letting each node uses local information only on its immediate neighbors. Consequently, one does not need a central entity that has full knowledge of the network topological structure. We experimentally evaluate the performance of the HellRank measure in correlation with other centrality measures on real-world networks. The results show partial ranking similarity between the HellRank and the other conventional metrics according to the Kendall and Spearman rank correlation coefficient.
\end{abstract}

\section{Introduction}
Social networks have become a very important social structure of our modern society with hundreds of millions of users nowadays. With the growth of information spread across various social networks, the question of \lq\lq how to measure the \textit{relative importance} of users in a social network?\rq\rq\ has become increasingly challenging and interesting, as important users are more likely to be infected by, or to infect, a large number of users. Understanding users' behaviors when they connect to social networking sites creates opportunities for richer studies of social interactions. Also, finding a subset of users to statistically represent the original social network is a fundamental issue in social network analysis. This small subset of users (the behaviorally-representative users) usually plays an important role in influencing the social dynamics on behavior and structure.

The centrality measure is widely used in social network analysis to quantify the \textit{relative importance} of nodes within a network. The most central nodes are often the nodes that have more weight, both in terms of the number of interactions as well as the number of connections to other nodes \cite{Silva2013}. In social network analysis, such a centrality notion is used to identify influential users \cite{Mahyar2015TopK,Silva2013,Wei2013,Yustiawan2015}, as the influence of a user is the ability to popularize a particular content in the social network. To this end, various centrality measures have been proposed over the years to rank the network nodes according to their topological and structural properties \cite{Cha2010,Friedkin2011,Zhao2013}. These measures can be considered as several points of view with different computational complexity, ranging from low-cost measures (e.g., Degree centrality) to more costly measures (e.g., Betweenness and Closeness centralities) \cite{Wehmuth2013,Technology2015}. The authors of \cite{Stephenson1989} concluded that centrality may not be restricted to shortest paths. In general, the global topological structure of many networks is initially unknown. However, all these structural network metrics require full knowledge of the network topology \cite{Wehmuth2013,Mahyar2015CScomdet}.

\begin{figure}[t]
  \centering
\includegraphics[width=3in]{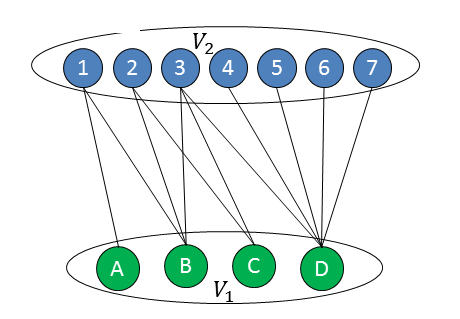} 
\vspace{-5pt}
 \caption{\small Bipartite graph $G=(V_1,V_2,E)$ with two different nodes set $V_1= \{A, B, C, D\}$, $ V_2= \{1,2,3,4,5,6,7\}$ and link set $E$ that each link connects a node in $V_1$ to a node in $V_2$.}
\label{fig:bipNet}
\vspace{-10pt}
\end{figure}

An interesting observation is that many real-world social networks have a bi-modal nature that allows the network to be modeled as a bipartite graph (see Figure \ref{fig:bipNet}). In a bipartite network, there are two types of nodes and the links can only connect nodes of different types \cite{Zhao2013}. The \textit{Social Recommender System} is one of the most important systems that can be modeled as a bipartite graph with users and items as the two types of nodes, respectively. In such systems, the centrality measure can have different interpretations from conventional centrality measures such as Betweenness, Closeness, Degree, and PageRank \cite{Kitsak2010}. The structural metrics, such as Betweenness and Closeness centrality, are known as the most common central nodes' identifier in one-mode networks, although in bipartite social networks they are not usually appropriate in identifying central users that are perfect representative for the bipartite network structure. For example, in a social recommender system \cite{mahyar2017recommender,taheri2017recommender}, that can be modeled by the network graph in Figure \ref{fig:bipNet}, user $D\in V_1$ is associated with items that have too low connections and have been considered less often by other users; meanwhile user $D$ is considered as the most central node based on these common centrality metrics, because it has more connections. However user $B \in V_1$ is much more a real representative than $D$ in the network. In the real-world example of an online store, if one user buys a lot of goods, but these goods are low consumption, and another buys fewer goods, but these are widely, we treat the second user as being a synechdochic representative of all users of the store. This is quite different from a conventional centrality metric outcome. 

Another interesting observation is that the common centrality measures are typically defined for non-bipartite networks. To use these measures in bipartite networks, different projection methods have been introduced to converting bipartite to monopartite networks \cite{Zhou2007,Sawant2013}. In these methods, a bipartite network is projected by considering one of the two node sets and if each pair of these nodes shares a neighbor in the network, two nodes will be connected in the projected one-mode network \cite{Latapy2008,Liebig2014}. One of the major challenges is that every link in a real network is formed independently, but this does not happen in the projected one-mode network. Because of lack of independency in the formation of links in the projected network, analysis of the metrics that use the random network \cite{Erdos1959} as a basis for their approach, is difficult. Classic random networks are formed by assuming that links are being independent from each other \cite{Opsahl2013}. The second challenge is that the projected bipartite network nodes tend to form \textit{Cliques}. A clique is a fully connected subset of nodes that all of its members are neighbors. As a result, the metrics that are based on triangles (i.e., a clique on three nodes) in the network, can be inefficient (such as structural holes or clustering coefficient measures) \cite{Wasserman1994,Opsahl2013}.

Despite the fact that the projected one-mode network is less informative than its corresponding bipartite representation, some of the measures for monopartite networks have been extended to bipartite mode \cite{Kitsak2011,Opsahl2013}. Moreover, because of requirement of full knowledge of network topology and lack of proper measure for detection of more behavioral representative users in bipartite social networks, the use of conventional centrality measures in the large-scale networks (e.g. in recommender systems) is a challenging issue. In order to overcome the aforementioned challenges and retain the original information in bipartite networks, proposing an accurate centrality measure in such networks seems essential \cite{Armour2014a,Liebig2014}.

Motivated by these observations and taking into account users' importance indicators for detection of central nodes in social recommender systems, we introduce a new centrality measure, called HellRank. This measure identifies central nodes in bipartite social networks. HellRank is based on the \textit{Hellinger distance}\cite{Nikulin2001}, a type of \textit{f-divergence} measure, that indicates structural similarity of each node to other network nodes. Hence, this distance-based measure is accurate for detecting the more behavioral representative nodes. We empirically show that nodes with high HellRank centrality measure have relatively high Degree, Betweenness and PageRank centrality measures in bipartite networks. In the proposed measure, despite of different objectives to identify central nodes, there is a partial correlation between HellRank and other common metrics.

The rest of the paper is organized as follows. In Section 2, we discuss related work on behavioral representative and influence identification mechanisms. We also discuss centrality measures for bipartite networks, and highlight the research gap between our objectives and previous work. In Section 3, we introduce our proposed measure to solve the problem of centrality in bipartite networks. Experimental results and discussions are presented in Section 4. We conclude our work and discuss the future works in Section 5.

\section{Related Work}

We organize the relevant studies on social influence analysis and the problem of important users in three different categories. First, in Section \ref{behavRepUsBip}, we study existing work on behavioral representative users detection methods in social networks. Second, in Section \ref{mecInflBip}, we review previous mechanisms for identifying influential users in social networks by considering the influence as a measure of the relative importance. Third, in Section \ref{centMeasBip}, we focus in more details on centrality measures for bipartite networks. 

\subsection{Behavioral Representative Users Detection}\label{behavRepUsBip}
Unlike influence maximization, in which the goal is to find a set of nodes in a social network who can maximize the spread of influence \cite{chen2009efficient,kempe2003maximizing}, the objective of behavioral representative users detection is to identify a few average users who can statistically represent the characteristics of all users \cite{landauer1988research}. Another type of related work is social influence analysis. 
\cite{anagnostopoulos2008influence} and \cite{singla2008yes} proposed methods to qualitatively measure the existence of influence. \cite{crandall2008feedback} studied the correlation between social similarity and influence. \cite{tang2009social}  presented a method for measuring the strength of such influence. The problem of sampling representative users from social networks is also relevant to graph sampling \cite{leskovec2006sampling,maiya2011benefits,ugander2013graph}. \cite{zhu2007improving} introduced a novel ranking algorithm called GRASSHOPPER, which ranks items with an emphasis on diversity. Their algorithm is based on random walks in an absorbing Markov chain. \cite{benevenuto2009characterizing} presented a comprehensive view of user behavior by characterizing the type, frequency, and sequence of activities users engage in and described representative user behaviors in online social networks based on clickstream data. \cite{giroire2008cubicle} found significant diversity in end-host behavior across environments for many features, thus indicating that profiles computed for a user in one environment yield inaccurate representations of the same user in a different environment. \cite{maia2008identifying} proposed a methodology for characterizing and identifying user behaviors in online social networks.

However, most existing work focused on studying the network topology and ignored the topic information. \cite{sun2013learning} aimed to find representative users from the information spreading perspective and \cite{ahmed2014network} studied the network sampling problem in the dynamic environment. \cite{papagelis2013sampling}  presented a sampling-based algorithm to efficiently explore a user's ego network and to quickly approximate quantities of interest. \cite{davoodi2012social} focused on the use of the social structure of the user community, user profiles and previous behaviors, as an additional source of information in building recommender systems. \cite{tang2015sampling} presented a formal definition of the problem of sampling representative users from social network.

\subsection{Identifying Influential Users}\label{mecInflBip}

\cite{Goyal2010} studied how to infer social probabilities of influence by developing an algorithm to scan over the log of actions of social network users using real data. \cite{Bharathi2007,Tang2009} focused on the influence maximization problem to model the social influence on large networks. TwitterRank, as an extension of PageRank metric, was proposed by \cite{Weng2010} to identify influential users in Twitter. \cite{Chen2013} used the Susceptible-Infected-Recovered (SIR) model to examine the spreading influence of the nodes ranked by different influence measures. \cite{Xu2012} identified influencers using joint influence powers through Influence network. \cite{Zhu2013} identified influencial users by using user trust networks. \cite{Li2014} proposed the weighted LeaderRank technique by replacing the standard random walk to a biased random walk. \cite{sun2014sparsity} presented a novel analysis on the statistical simplex as a manifold with boundary and applied the proposed technique to social network analysis to rank a subset of influencer nodes. \cite{Tang2012} proposed a new approach to incorporate users' reply relationship, conversation content and response immediacy to identify influential users of online health care community. \cite{Du2014} used multi-attribute and homophily characteristics in a new method to identify influential nodes in complex networks.

In the specific area of identifying influential users in bipartite networks, \cite{Beguerisse-Diaz2009} presented a dynamical model for rewiring in bipartite networks and obtained time-dependent degree distributions. \cite{Liebig2014} defined a bipartite clustering coefficient by taking differently structured clusters into account, that can find important nodes across communities. The concept of clustering coefficient will be discussed in further detail in the Section \ref{CC}.

\subsubsection{Clustering Coefficient}\label{CC}
This measure shows the nodes' tendency to form clusters and has attracted a lot of attention in both empirical and theoretical work. In many real-world networks, especially social networks, nodes are inclined to cluster in densely connected groups \cite{Opsahl2009}. Many measures have been proposed to examine this tendency. In particular, the global clustering coefficient provides an overall assessment of clustering in the network \cite{Luce1949}, and the local clustering coefficient evaluates the clustering value of the immediate neighbors of a node in the network \cite{Watts1998}.

The global clustering coefficient is the fraction of 2-paths (i.e., three nodes connected by two links) that are closed by the presence of a link between the first and the third node in the network. The local clustering coefficient is the fraction of the links among a node's interactions over the maximum possible number of links between them \cite{Watts1998,Opsahl2013}. 

Due to structural differences, applying these general clustering coefficients directly to a bipartite network, is clearly not appropriate \cite{Borgatti1997}, Thus the common metrics were extended or redefined, and different clustering measures were defined in these networks. In one of the most common clustering coefficients in bipartite networks, 4-period density is measured instead of triangles \cite{Robins2004,Zhang2008}. However, this measure could not consider the triple closure concept in the clustering as it actually consists of two nodes. This kind of measure can only be a measure of the level of support between two nodes rather than the clustering of a group of nodes. Accordingly, \cite{Latapy2008}  defined the notion of clustering coefficient for pairs of nodes capturing correlations between neighborhoods in the bipartite case. Addtionally, \cite{Opsahl2013} considered the factor, $C^*$, as the ratio of the number of closed 4-paths ($\tau^*_\Delta$) to the number of 4-paths ($\tau^*$), as:
\begin{eqnarray}
\centering
C^* =\frac{ {\rm closed\ 4\textendash paths}}{{\rm 4\textendash paths}}=\frac{\tau^*_\Delta}{\tau^*}
\end{eqnarray}

\subsection{Centrality Measures for Bipartite Networks}\label{centMeasBip}
Various definitions for centrality have been proposed in which centrality of a node in a network is generally interpreted as the relative importance of that node \cite{freeman1978centrality,Chen2013}. Centrality measures has attracted a lot of attentions as a tool to analyze various kinds of networks (e.g. social, information, and biological networks) \cite{Faust1997,Kang2011}. In this section, we consider a set of well-known centrality measures including Degree, Closeness, Betweenness, Eigenvector and PageRank, all of them redefined for bipartite networks. Given bipartite network $G=(V_1,V_2,E)$, where $V_1$ and $V_2$ are the two sides of network with $|V_1|=n_1$ and $|V_2|=n_2$. The link set $E$ includes all links connecting nodes of $V_1$ to nodes of $V_2$. For the network in Figure \ref{fig:bipNet}, $n_1$ and $n_2$ are equal to 4 and 7, respectively. Let $Adj$ be the adjacency matrix of this network, as shown below:
\begin{eqnarray}
\centering
 Adj(G)=
\begin{blockarray}{cccccccc}
&1 & 2 & 3 & 4 & 5 & 6 & 7\\
\begin{block}{c(ccccccc)}
  A & 1 & 0 & 0 & 0 & 0 & 0 & 0 \\
  B & 1 & 1 & 1 & 0 & 0 & 0 & 0 \\
  C & 0 & 1 & 1 & 0 & 0 & 0 & 0 \\
  D & 0 & 0 & 1 & 1 & 1 & 1 & 1 \\
\end{block}
\end{blockarray}
\end{eqnarray}

\subsubsection{Degree Centrality}\label{Degree}
In one-mode graphs, Degree centrality of node $i$, $d_i$, is equal to the number of connections of that node. In bipartite graphs, it indicates the number of node's connections to members on the other side. For easier comparison, Degree centrality is normalized: the degree of each node is divided by the size of the other node set. Let $d_i^*$ be the normalized Degree centrality of node $i$. This is equal to \cite{Borgatti1997,Faust1997}:
\begin{eqnarray}
d^*_i = \frac{d_i}{n_2} ,\ d^*_j = \frac{d_j}{n_1};\ \ \  i\in V_1,  j\in V_2
\end{eqnarray}

As the network size becomes increasingly large, employing Degree centrality is the best option \cite{You2015}. On the other hand, this centrality is based on a highly local view around each node. As a consequence, we need more informative measures that can further distinguish among nodes with the same degree \cite{Kang2011}.

\subsubsection{Closeness Centrality}\label{Closeness}
The standard definition of Closeness $c_i$ for node $i$ in monopartite networks, refers to the sum of geodesic distances from node $i$ to all $n-1$ other nodes in the network with $n$ nodes \cite{sabidussi1966centrality}. For bipartite networks, this measure can be calculated using the same approach, but the main difference is normalization. Let $c_i^*$ be the normalized Closeness centrality of node $i \in V_1$. This is equal to \cite{Borgatti1997,Faust1997}:
\begin{eqnarray}
c_i^*= \frac{n_2+2(n_1-1)}{c_i};\ \ \  i\in V_1\\c_j^*= \frac{n_1+2(n_2-1)}{c_j};\ \ \   j\in V_2 
\end{eqnarray}

For the bipartite network shown in Figure \ref{fig:bipNet}, normalized Closeness centrality of the nodes A, B, C and D are respectively equal to 0.35, 0.61, 0.52 and 0.68. It specifies that node D is the most central node which says that Closeness centrality cannot help us very much in the objective to finding more behavioral representative nodes in bipartite social networks.

\subsubsection{Betweenness Centrality}\label{Betweenness}
Betweenness centrality of node $i$, $b_i$, refers to the fraction of shortest paths in the network that pass through node $i$ \cite{freeman1977set}. In bipartite networks, maximum possible Betweenness for each node is limited by relative size of two nodes sets, as introduced by \cite{Borgatti2011}:
\begin{eqnarray}
b_{\max}(V_1)&=&\frac{1}{2}[n_2^2(s+1)^2 + \nonumber \\
&\ & n_2(s+1)(2t-s-1)-t(2s-t+3)]
\end{eqnarray}

\noindent where $s=(n_1-1)\  div \ n_2$ and $t=(n_1-1)\  mod \ n_2$; and
\begin{eqnarray}
b_{\max}(V_2)&=&\frac{1}{2}[n_1^2(p+1)^2 + \nonumber \\
&\ & n_1(p+1)(2r-p-1)-r(2p-r+3)]
\end{eqnarray}

\noindent where $p=(n_2-1)\  div \ n_1$ and $r=(n_2-1)\  mod \ n_1$.

For the bipartite network shown in Figure \ref{fig:bipNet}, normalized Betweenness centrality of the nodes A, B, C and D are respectively equal to 0, 0.45, 0.71 and 0.71. It specifies that nodes C and D are the most central nodes which says that Betweenness centrality cannot help us much objective in finding more behavioral representative nodes in bipartite social networks.

\subsubsection{Eigenvector and PageRank Centrality}\label{PageRank}
Another important centrality measure is Eigenvector centrality, which is defined as the principal eigenvector of adjacency matrix of the network. A node's score is proportional to the sum of the scores of its immediate neighbors. This measures exploits the idea that nodes with connections to high-score nodes are more central \cite{Bonacich1972}. The eigenvector centrality of node $i$, $e_i$, is defined as follows \cite{Faust1997}:
\begin{eqnarray}
e_i = \lambda \sum a_{ij}e_j
\end{eqnarray}
where $Adj(G)=(a_{ij})_{i,j=1}^n$ denotes the adjacency matrix of the network with $n$ nodes and $\lambda$ is the principal eigenvalue of the adjacency matrix. Our interest regarding to the eigenvector centrality is particularly focused on the distributed computation of PageRank \cite{You2015}. The PageRank is a special case of eigenvector centrality when the adjacency matrix is suitably normalized to obtain a column stochastic matrix \cite{You2015}. The PageRank vector $R=(r_1,r_2,\dots,r_n)^T$ is the solution of the following equation:
\begin{eqnarray}
R=\frac{1-d}{n}.1 + dLR 
\end{eqnarray}
where $r_i$ is the PageRank of node $i$ and $n$ is the total number of nodes. $d$ is a damping factor, set to around 0.85 and $L$ is a modified adjacency matrix, such that $l_{i,j} = 0$ if and only if node $j$ does not have a link to $i$ and $\sum_{i=1}^n l_{i,j}=1$, where $l_{i,j}=\frac{a_{i,j}}{d_j}$, and $d_j=\sum_{i=1}^n a_{i,j}$ is the out-degree of node $j$ \cite{Okamoto2008}. 

For the bipartite network shown in Figure \ref{fig:bipNet}, normalized PageRank centrality of the nodes A, B, C and D are respectively equal to 0.05, 0.11, 0.08 and 0.21. It specifies that node D is the most central node which says that PageRank centrality cannot help us much objective to finding more behavioral representative nodes in bipartite social networks.

\section{Proposed Method}
In this paper, we want to identify more behavioral representative nodes in bipartite social networks. To this end, we propose a new similarity-based centrality measure, called HellRank. Since the similarity measure is usually inverse of the distance metrics, we first choose a suitable distance measure, namely Hellinger distance (Section \ref{choose_metric}). Then we apply this metric to bipartite networks. After that, we theoretically analyze the impact of the distance metric in the bipartite networks. Next, we generate a distance matrix on one side of the network. Finally, we compute the HellRank score of each node, accordingly to this matrix. As a result, the nodes with high HellRank centrality are more behavioral representative nodes in bipartite social networks.

\subsection{Select A Well-Defined Distance Metric}\label{choose_metric}
When we want to choose a base metric, an important point is whether this measure is based on a well-defined mathematical metric. We want to introduce a similarity-based measure for each pair of nodes in the network. So, we choose a proper distance measure as base metric, because the similarity measures are in some sense the inverse of the distance metrics. A true distance metric must have several main characteristics. A metric with these characteristics on a space induces topological properties (like open and closed sets). It leads to the study of more abstract topological spaces. \cite{hunter2012} introduced the following definition for a distance metric.

\begin{definition}\label{well-defined metric}
A metric space is a set $X$ that has a notion of the \textit{distance function} $d(x, y)$ between every pair of points $x, y \in X$. A \textit{well-defined distance metric} $d$ on a set $X$ is a function $d : X \times X \rightarrow {\rm I\!R}$ such that for all $x, y,z \in X$, three properties hold:
\begin{enumerate}
\item \textit{Positive Definiteness}: $d(x, y) \geq 0$ and $d(x, y) = 0$ if and only if $x = y$;
\item \textit{Symmetry}: $d(x, y) = d(y, x)$;
\item \textit{Triangle Inequality}: $d(x,y) \leq d(x,z) + d(z,y)$.
\end{enumerate}
\end{definition}

We define our distance function as the difference between probability distribution for each pair of nodes based on f\textnormal{-}divergence function, which is defined by \cite{csiszar2004}:
\begin{definition}
An \textit{f\textnormal{-}divergence} is a function $D_f(P||Q) $ that measures the difference between two probability distributions $P$ and $Q$. For a convex function $f$ with $f(1) = 0$, the \textit{f\textnormal{-}divergence} of $Q$ from $P$ is defined as:
\begin{eqnarray}
D_f(P||Q)&=& \int_{\Omega} {f(\frac{dP}{dQ})dQ}
\end{eqnarray}
\end{definition}
where $\Omega$ is a sample space, which is the set of all possible outcomes.
 
In this paper, we use one type of the \textit{f-divergence} metric, called \textbf{Hellinger distance} (aka \textit{Bhattacharyya distance}), that was introduced by \textit{Ernst Hellinger} in 1909 \cite{Nikulin2001}. In probability theory and information theory, \textbf{Kullback-Leibler divergence} \cite{Kullback1951} is a more common measure of difference between two probability distributions, however it does not satisfy both the symmetry and the triangle inequality conditions \cite{VanderVaart1998}. Thus, this measure is not intuitively appropriate to explain similarity in our problem. As a result, we choose Hellinger distance to quantify the similarity between two probability distributions \cite{VanderVaart1998}. For two discrete probability distributions $P=(p_1, \ldots, p_m) $ and $Q=(q_1, \ldots, q_m)$, in which $m$ is length of the vectors, Hellinger distance is defined as:
\begin{eqnarray}
D_H(P||Q)&=& \frac{1}{\sqrt{2}} \sqrt{\sum_{i=1}^{m}(\sqrt{p_i}-\sqrt{q_i})^2}
\end{eqnarray}

It is obviously related to the Euclidean norm of the difference of the square root of vectors, as:
\begin{eqnarray}
D_H(P||Q)&=& \frac{1}{\sqrt{2}} \Vert{\sqrt{P}-\sqrt{Q}}\Vert_2
\end{eqnarray}

\subsection{Applying Hellinger distance in Bipartite Networks}
In this Section, we want to apply the Hellinger distance to a bipartite network for measuring the similarity of the nodes on one side of the network. Assume $x$ is a node in a bipartite network in which its neighborhood is $N(x)$ and its degree is $deg(x)=|N(x)|$. Suppose that the greatest node degree of the network is $\Delta$. Let $l_i$ be the number of $x$'s neighbors with degree of $i$. Suppose the vector $L_x=(l_1,\dots,l_{\Delta})$ be the non-normalized distribution of $l_i$ for all adjacent neighbors of $x$. Now, we introduce the Hellinger distance between two nodes $x$ and $y$ on one side of the bipartite network as follows:  
\begin{eqnarray}
d(x,y) = \sqrt{2}\ D_H (L_x\|L_y)
\end{eqnarray}
The function $d(x,y)$ represents the difference between two probability distribution of $L_x$ and $L_y$. To the best of our knowledge, this is the first work that introduces the Hellinger distance between each pair of nodes in a bipartite network, using degree distribution of neighbors of each node.

\subsection{Theoretical Analysis}
In this Section, we first express the Hellinger distance for all positive real vectors to show that applying this distance to bipartite networks still satisfies its metricity (lemma \ref{positive_metric}) according to Definition \ref{well-defined metric}. Then, we find an upper and a lower bound for the Hellinger distance between two nodes of bipartite network.

\begin{lemma}\label{positive_metric}
Hellinger distance for all positive real vectors is a well-defined distance metric function.
\end{lemma}

\begin{proof}
Based on the true metric properties in Definition \ref{well-defined metric}, for two probability distribution vectors $P$ and $Q$, the following holds:
\begin{eqnarray}
D_H(P\|Q) &\geq&0 \\
D_H(P\|Q) &=&0 \ \ \Leftrightarrow \ P=Q \\
D_H(P\|Q)  &=& D_H(Q\|P)  
\end{eqnarray}

If we have another probability distribution $R$ similar to $P$ and $Q$, then according to the triangle inequality in norm 2, we should have: 
\begin{eqnarray}
\frac{1}{\sqrt{2}} \|{\sqrt{P}-\sqrt{Q}}\|_2&\leq&  \frac{1}{\sqrt{2}} ( \|{\sqrt{P}-\sqrt{R}}\|_2+ \|{\sqrt{R}-\sqrt{Q}}\|_2 ) \nonumber\\ 
\Rightarrow~~ D_H(P\|Q)&\leq& D_H(P\|R)+ D_H(R\|Q) 
\end{eqnarray}

It shows that the triangle inequality in Hellinger distance for all positive real vectors is a well-defined distribution metric function.
\end{proof}

Using this distance measure, we have the ability to detect differences between local structures of nodes. In other words, this distance expresses similarity between the local structures of two nodes. If we normalize the vectors (i.e., sum of the elements equals to one), then differences between local structures of nodes may not be observed. For example, there does not exist any distance between node $x$ with $deg(x)=10$ that its neighbors' degree are 2 and node $y$ with $deg(y)=1$ that its neighbor's degree are 2. Therefore, our distance measure with vectors normalization is not proper for comparing two nodes.

Then, we claim that if the difference between two nodes' degree is greater (or smaller) than a certain value, the distance between these nodes cannot be less (or more) than a certain value. In other words, their local structures cannot be similar more (or less) than a certain value. In the following theorem, we find an upper and a lower bound for the Hellinger distance between two nodes on one side of a bipartite network using their degrees' difference.

\begin{theorem}
If we have two nodes $x$ and $y$ on one side of a bipartite network, such that $deg(x)=k_1$, $deg(y)=k_2$, and $k_1 \geq k_2$, then we have a lower bound for the distance between these nodes as:
\begin{eqnarray}
d(x,y)\geq \sqrt{k_1}-\sqrt{k_2}
\end{eqnarray}
and an upper bound as: 
\begin{eqnarray}
d(x,y)\leq \sqrt{k_1+k_2}
\end{eqnarray}
\end{theorem}

\begin{proof}
To prove the theorem, we use the Lagrange multipliers. Suppose $L_x=(l_1,\dots,l_{\Delta})$ and $L_y=(h_1,\dots,h_{\Delta})$ are positive real distribution vectors of nodes $x$ and $y$. Based on (13) we know $d(x,y) = \sqrt{2}\ D_H (L_x\|L_y)$, so one can minimize the distance between these nodes by solving $\min\limits_{L_x,L_y} \sqrt{2}\ D_H(L_x\|L_y)$, which is equivalent to find the minimum square of their distance:
\begin{eqnarray}
\min_{L_x,L_y} 2 D_H^2(L_x \| L_y)=\min_{L_x,L_y} \sum_{i=1}^{\Delta}(\sqrt{l_i}-\sqrt{h_i})^2 \nonumber
\end{eqnarray}
So, Lagrangian function can be defined as follows:
\begin{eqnarray}
F(L_x,L_y,\lambda_1,\lambda_2) &=&{\sum_{i=1}^{\Delta} (\sqrt {l_{i}}-\sqrt {h_{i}})^2} +\nonumber \\
&\ &\lambda_1(k_1-\sum_{i=1}^{\Delta} {l_{i}})+ \lambda_2(k_2-\sum_{i=1}^{\Delta} {h_{i}}) \nonumber
\end{eqnarray}
Then, we take the first derivative with respect to $l_i$:
\begin{eqnarray}
\frac{\partial F}{\partial l_{i}} = 1 - \frac{\sqrt{h_{i}}}{\sqrt {l_{i}}}- \lambda_1 =0 ~~~\Rightarrow~~~ h_{i} = {l_{i}}{(1-\lambda_1)^2} \nonumber
\end{eqnarray}
Due to $\sum_{i=1}^{\Delta}l_i=k_1$ and $\sum_{i=1}^{\Delta}h_i=k_2$, we have:
\begin{eqnarray}
\sum_{i=1}^{\Delta}h_{i} = k_2 \rightarrow  \sum_{i=1}^{\Delta} {l_{i}}{(1-\lambda_1)^2} = k_2 \rightarrow (1-\lambda_1)  =  \pm \sqrt{\frac{k_2}{k_1}} \nonumber
\end{eqnarray}
But in order to satisfy $\sqrt{h_i}= {\sqrt{l_i}}{(1-\lambda_1)}$, the statement $1-\lambda_1$ must be positive, thus:
\begin{eqnarray}
h_{i} ={l_{i}}{(1-\lambda_1)^2} = l_{i} \frac {k_2}{k_1}  \nonumber
\end{eqnarray}
After derivation with respect to $h_i$, we also reach similar conclusion. If this equation is true, then equality statement for minimum function will occur, as:
\begin{eqnarray}
\min_{L_x,L_y} 2 D_H^2(L_x \| L_y)&=& \sum_{i=1}^{\Delta}(\sqrt{l_i}-\sqrt{{\frac{k_2}{k_1}}}\sqrt{l_i})^2 \nonumber\\
&=&\sum_{i=1}^{\Delta}l_i(1-\sqrt{{\frac{k_2}{k_1}}})^2\nonumber\\ 
&=&(1-\sqrt{{\frac{k_2}{k_1}}})^2\sum_{i=1}^{\Delta}l_i\nonumber\\ 
&=&(1-\sqrt{{\frac{k_2}{k_1}}})^2 k_1 \ \ ~~~~~~~\Rightarrow   \nonumber\\ 
\min_{L_x,L_y} \sqrt{2}\ D_H(L_x \| L_y)&=&\sqrt{k_1}(1-\sqrt{{\frac{k_2}{k_1}}}) \nonumber \\ 
&=& \sqrt{k_1}-\sqrt{k_2}\nonumber
\end{eqnarray}

So, the lower bound for distance of any pair of nodes on one side of the bipartite network could not be less than a certain value by increasing their degrees difference.

Now, we want to find an upper bound according to Equation (19). As we know, the following statement is true for any $p_{i}, p_{j}$ and $q_{i}, q_{j}$:
\begin{eqnarray}
(\sqrt{p_i}-\sqrt{q_i})^2&+&(\sqrt{p_j}-\sqrt{q_j})^2 \nonumber\\
&\leq&(\sqrt{p_i+p_j}-0)^2+(\sqrt{q_i+q_j}-0)^2\nonumber\\
&=& p_i+p_j+q_i+q_j \nonumber
\end{eqnarray}
Suppose in our problem, $p_i=l_i$, $p_j=l_j$, and $q_i=h_i$, $q_j=h_j$, then this inequality holds for any two pairs of elements in $L_x$ and $L_y$. Eventually we have:
\begin{eqnarray}
d(x,y) \leq \sqrt{(\sqrt{k_1}-0)^2+(\sqrt{k_2}-0)^2} = \sqrt{k_1+k_2}\nonumber
\end{eqnarray}
We can conclude that it is not possible for any pair of nodes on one side of the bipartite network that their distance to be more than a certain value by increasing their degrees.
\end{proof}

As a result, we found the upper and the lower bounds for the Hellinger distance between two nodes on one side of the bipartite network using their degrees' difference.

%Now we analyze the time complexity of our method.

%\subsubsection{Time Complexity}\hspace*{\fill}\\
%Calculation of non-normalized degrees distribution of nodes and update nodes vectors on two sides of each link (degree of nodes), both takes $O(m)$ time, where $m$ is the number of links in the graph. Also, finding Hellinger distance of the nodes takes $O(n^2\Delta)$ time, that $n$ is the number of all nodes in one part of the graph and $\Delta$ is the length of the $L_x$ vector. Then in overall, obtaining Hellinger distance matrix requires $O(m+n^2\Delta)$ time. 

\subsubsection{An Example with Probabilistic View}
In this example, we want to analyze the similarity among nodes based on Hellinger distance information in an artificial network. We examine how we can obtain required information for finding similar nodes to a specific node $x$ as the expected value and variance of the Hellinger distance. Suppose that in a bipartite artificial network with $\vert V_1\vert = n_1$ nodes on one side and $\vert V_2\vert = n_2$ nodes on the other side, nodes in $V_1$ is connected to nodes in $V_2$ using Erd\"{o}s-R\'{e}nyi model $G(m, p)$. In other words, there is an link with probability $p$ between two sides of the network. Distribution function $L_x$ of node $x\in V_1$ can be expressed as a Multinomial distribution-form as:
\begin{eqnarray}
P\left(l_1,\dots,l_{\Delta}|deg(x)=k \right) &=&P\left(L_x|deg(x)=k\right) \nonumber\\
 &=&\left(\begin{array}{c} k \\ l_1,\dots,l_{\Delta} \end{array}\right)\prod{P_i^{l_i}} 
\end{eqnarray}
where $P_i=\left(\begin{array}{c} n_2-1 \\ i-1 \end{array}\right)p^{i-1} (1-p)^{n_2-i}$ is a Binomial distribution probability $B(n_2, p)$ for $x$'s neighbors that their degree is equal to $i$.

According to the Central Limit Theorem \cite{johnson2004information}, Binomial distribution converges to a Poisson distribution $Pois(\lambda)$ with parameter $\lambda=(n_2-1)p$ and the assumption that $(n_2-1)p$ is fixed and $n_2$ increases. Therefore, average distribution of $P\left(L_x|deg(x)=k\right)$ will be $\mu=(kp_1,kp_2,\dots,kp_{\Delta})$. In addition, degree distribution in Erd\"{o}s-R\'{e}nyi model converges to Poisson distribution by increasing $n_1$ and $n_2$ ($\lambda=n_1 p$ for one side of network and $\lambda=n_2 p$ for another one).

The limit of average distribution of $P\left(L_x|deg(x)=k\right)$ by increasing $\Delta$, approaches $k$ times of a Poisson distribution. Thus, normalized $L_x$ vector is a Poisson distribution with parameter $\lambda=(n_2-1)p$. To find a threshold for positioning similar and closer nodes to node $x$, we must obtain expectation and variance of the Hellinger distance between $x$ and the other nodes in node set $V_1$. Before obtaining these values, we mention the following lemma to derive equal expression of Hellinger distance and difference between typical mean and geometric mean.
\begin{lemma}
Suppose two distribution probability vectors \\$P=(p_1, \ldots, p_m) $ and $Q=(q_1, \ldots, q_m)$ that $P$ is $k_1$ times of a Poisson distribution probability vector $P_1\,\sim \,\rm{Poisson}(\lambda_1)$ and $Q$ is $k_2$ times of a Poisson distribution probability vector $P_2\,\sim \,\rm{Poisson}(\lambda_2)$\footnote{Vector P=$(p_0,p_1,\dots)$ is a Poisson distribution probability vector such that the probability of the random variable with Poisson distribution being $i$ is equal to $p_i$.}. The square of Hellinger distance between $P$ and $Q$ is calculated by:
\begin{eqnarray}
D_H^2(P \| Q) = \frac{k_1+k_2}{2} - \sqrt{k_1k_2}(1-e^{-\frac{1}{2}(\sqrt{\lambda_1}-\sqrt{\lambda_2})^2})  
\end{eqnarray}
\end{lemma}
\begin{proof}
The squared Hellinger distance between two Poisson distributions $P_1$ and $P_2$ with rate parameters $\lambda_1$ and $\lambda_2$ is \cite{torgersen1991}:
 \begin{eqnarray}
D_H^2(P_1\|P_2)= 1-e^{-\frac{1}{2}(\sqrt{\lambda_1}-\sqrt{\lambda_2})^2}
\end{eqnarray}

Therefore,  the squared Hellinger distance for probability vectors $P$ and $Q$, will be equal to $(\sum_{i=1}^m p_i= k_1,\\\sum_{i=1}^m q_i=k_2)$:
\begin{eqnarray}
D_H^2(P \| Q) &=&\frac{1}{2} \sum_{i=1}^m(\sqrt{p_i}-\sqrt{q_i})^2 \nonumber\\
&=& \frac{1}{2} \sum_{i=1}^m(p_i+q_i - 2\sqrt{p_i q_i}) \nonumber\\
&=& \frac{k_1+k_2}{2} - \sqrt{k_1k_2}(1-e^{-\frac{1}{2}(\sqrt{\lambda_1}-\sqrt{\lambda_2})^2}) 
\end{eqnarray}
\end{proof}

However, in the special case of $\lambda_1=\lambda_2$, we have:
\begin{eqnarray}
D_H^2(P\|Q) = \frac{k_1+k_2}{2} - \sqrt{k_1k_2}
\end{eqnarray}
It means that the squared Hellinger distance is equal to difference between typical mean and geometric mean.\\

To calculate the second moment of distance between node $x \in V_1$ and any other nodes $z \in V_1$ in the same side of the bipartite network based on the lemma 2, we have:
\begin{eqnarray}
E_{z\in V_1}\left[d^2(x,z)\right]&=&E\left[2\ D_H^2(L_x\|L)\right]\nonumber\\
&=& \sum_{i=1}^\infty \left(\frac{e^{-n_1p}(n_1p)^i}{(n_1p)!}(k+i-2\sqrt{ki}) \right) \nonumber\\
&\simeq& \sum_{i=1}^{n_2} \left(\frac{e^{-n_1p}(n_1p)^i}{(n_1p)!}(k+i-2\sqrt{ki}) \right)\nonumber\\
&&
\end{eqnarray}
Where $L=(L_z|z \in V_1)$ and the infinite can be approximated by $n_2$ elements. Similarly, for distance expectation we have:
\begin{eqnarray}
E\left[\sqrt{2}\ D_H(L_x\|L)\right]&\simeq  \displaystyle\sum\limits_{i=1}^{n_2} \left(\frac{e^{-n_1p}(n_1p)^i}{(n_1p)!}\sqrt{(k+i-2\sqrt{ki}) }\right) \nonumber\\
&
\end{eqnarray}
In addition, variance can also be obtained based on these calculated moments:
\begin{eqnarray}
Var_{z\in V_1}\left(d(x,z)\right) &= E_{z\in V_1}\left[d^2(x,z)\right] - \left( E_{z\in V_1}\left[d(x,z)\right] \right)^2\nonumber\\
&
\end{eqnarray}

Hence, using these parameters, the required threshold for finding similar nodes to a specific node $x$, can be achieved. If we want to extend our method to more complex and realistic networks, we can assume that distribution $L_x$ is a multiple of Poisson distribution (or any other distribution) vector with parameter $\lambda_x$, in which $\lambda_x$ can be extracted by either the information about structure of the network or appropriate maximum likelihood estimation for node $x$. Therefore, the threshold will be more realistic and consistent with the structure of the real-world networks.

\subsubsection{Generalization to Weighted Bipartite Networks}
The introduced distance metric function can be extended to weighted networks. The generalized Hellinger distance between two nodes of the wighted bipartite network can be considered as:
\begin{eqnarray}
d(x,y) = \sqrt{2} D_H(W_x\|W_y)
\end{eqnarray}
where $W_x=(w'_1,\dots,w'_{\Delta})$, $w'_i = \sum\limits_{\substack{j \in N(x)\\deg(j)=i}}{w_j}$, and $w_j$ is the vector of weights on the links of the network. 

\subsection{ٍRank Prediction via HellRank}
In this Section, we propose a new Hellinger-based centrality measure, called HellRank, for the bipartite networks. Now, according to the Section 3.2, we find the Hellinger distances between any pair of nodes in each side of a bipartite network. Then we generate an $n_1\times n_1$ distance matrix ($n_1$ is the number of nodes in one side of network). The Hellinger distance matrix of $G$ shown in Figure \ref{fig:bipNet} is as follows:
\begin{eqnarray}
\centering
Hell\text{-} Matrix(G)=
 \begin{blockarray}{ccccc}
& A & B & C & D \\
\begin{block}{c(cccc)}
  A & 0 & 0.42 & 0.54 & 1 \\
  B & 0.42 & 0 & 0.12 & 0.86 \\
  C & 0.54 & 0.12 & 0 & 0.82 \\
  D & 1 & 0.86 & 0.82 & 1 \\
\end{block} 
\end{blockarray}\nonumber
 \end{eqnarray}

According to the well-defined metric features (in Section 3.1) and the ability of mapping to Euclidean space, we can cluster nodes based on their distances. It means that any pair of nodes in the matrix with a less distance can be placed in one cluster by specific neighborhood radius. By averaging inverse of elements for each row in the distance matrix, we get final similarity score (HellRank) for each node of the network, by:
\begin{eqnarray}
HellRank(x) &=& \frac {n_1} {\sum_{z \in V_1}{d(x,z)}} 
\end{eqnarray}
Let $HellRank^* (x)$ be the normalized HellRank of node $x$ that is equal to:
\begin{eqnarray}
HellRank^* (x) &=& HellRank(x) . \min_{z \in V_1}\left({HellRank(z)}\right)\nonumber
\end{eqnarray}
where $\min_{z \in V_1}\left({HellRank(z)}\right)$ is the minimum possible HellRank for each node

A similarity measure is usually (in some sense) the inverse of a distance metric: they take on small values for dissimilar nodes and large values for similar nodes. The nodes in one side with higher similarity scores represent more behavioral representation of that side of the bipartite network. In other words, these nodes are more similar than others to that side of the network. HellRank actually indicates structural similarity for each node to other network nodes. For the network shown in Figure \ref{fig:bipNet}, according to Hellinger distance matrix, normalized HellRank of nodes A, B, C, and D are respectively equal to 0.71, 1, 0.94, and 0.52. It is clear that among all of the mentioned centrality measures in Section 2.2, only HellRank considers node \lq B\rq\ as a more behavioral representative node. Hence, sorting the nodes based on their HellRank measures will have a better rank prediction for nodes of the network. The nodes with high HellRank is more similar to other nodes. In addition, we find nodes with less scores to identify very specific nodes which are probably very different from other nodes in the network. The nodes with less HellRank are very dissimilar to other nodes on that side of the bipartite network.

\section{Experimental Evaluation}
In this Section, we experimentally evaluate the performance of the proposed HellRank measure in correlation with other centrality measures on real-world networks. After summarizing datasets and evaluation metrics used in the experiments, the rest of this section addresses this goal. Finally, we present a simple example of mapping the Hellinger distance matrix to the Euclidean space to show clustering nodes based on their distances.

\subsection{Datasets}
To examine a measure for detection of central nodes in a two-mode network, South Davis women \cite{Davis1941}, is one of the most common bipartite datasets. This network has a group of women and a series of events as two sides of the network. A woman linked to an event if she presents at that event. Another data set used in the experiments is OPSAHL-collaboration network \cite{Newman2001}, which contains authorship links between authors and publications in the arXiv condensed matter Section (cond-mat)  with 16726 authors and 22015 articles. A link represents an authorship connecting an author and a paper.

\subsection {Evaluation Metrics}
One of the most popular evaluation metrics for comparison of different node ranking measures is Kendall's rank correlation coefficient ($\tau$). In fact, Kendall is non-parametric statistic that is used to measure statistical correlation between two random variables \cite{Abdi1955}:
\begin{eqnarray}
\tau = \frac{N_{<concordant \ pairs>} - N_{<discordant \ pairs>} }{\frac{1}{2} n(n-1)}
\end{eqnarray}
where $N_{<S>}$ is the size of set $S$.

Another way to evaluate ranking measures is binary vectors for detection of top-$k$ central nodes. All of vector's elements are zero by default and only top-$k$ nodes' values are equal to 1. To compare ranking vectors with the different metrics, we use Spearman's rank correlation coefficient ($\rho$) that is a non-parametric statistics to measure the correlation coefficient between two random variables \cite{Lehamn2005}:
\begin{eqnarray}
\rho = \frac{\sum_i (x_i-\overline{x})(y_i-\overline{y}) }{\sqrt{\sum_i (x_i-\overline{x})^2 \sum_i (y_i-\overline{y})^2}}
\end{eqnarray}
where $x_i$ and $y_i$ are ranked variables and $\overline{x}$ and $\overline{y}$ are mean of these variables.

\begin{figure*}[t]
\centering
\includegraphics[width=6in]{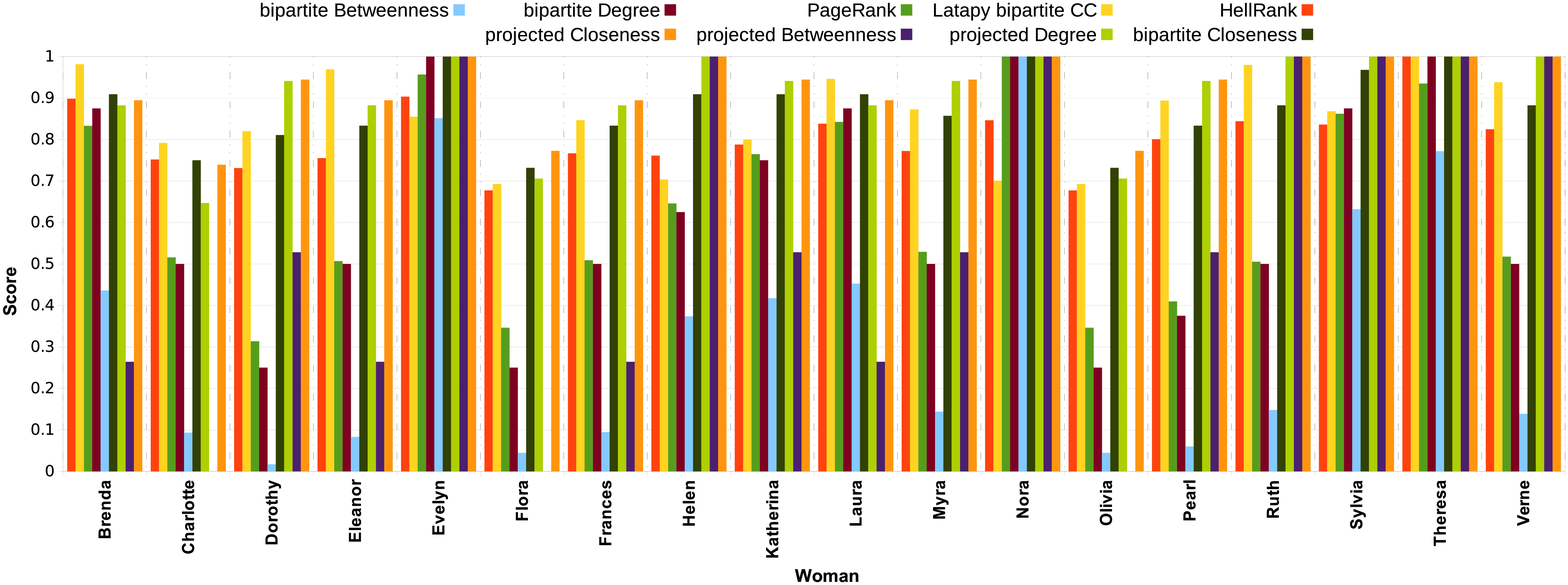}
\vspace{-3pt}
\caption{\small Comparison between rankings for all women in the Davis dataset based on various centrality metrics.}
\vspace{-5pt}
\label{fig:compares}
\end{figure*}

\begin{figure*}[!t]
\centering
\subfloat[\small HellRank vs. bipartite Betweenness]{\includegraphics[width=2.1in]{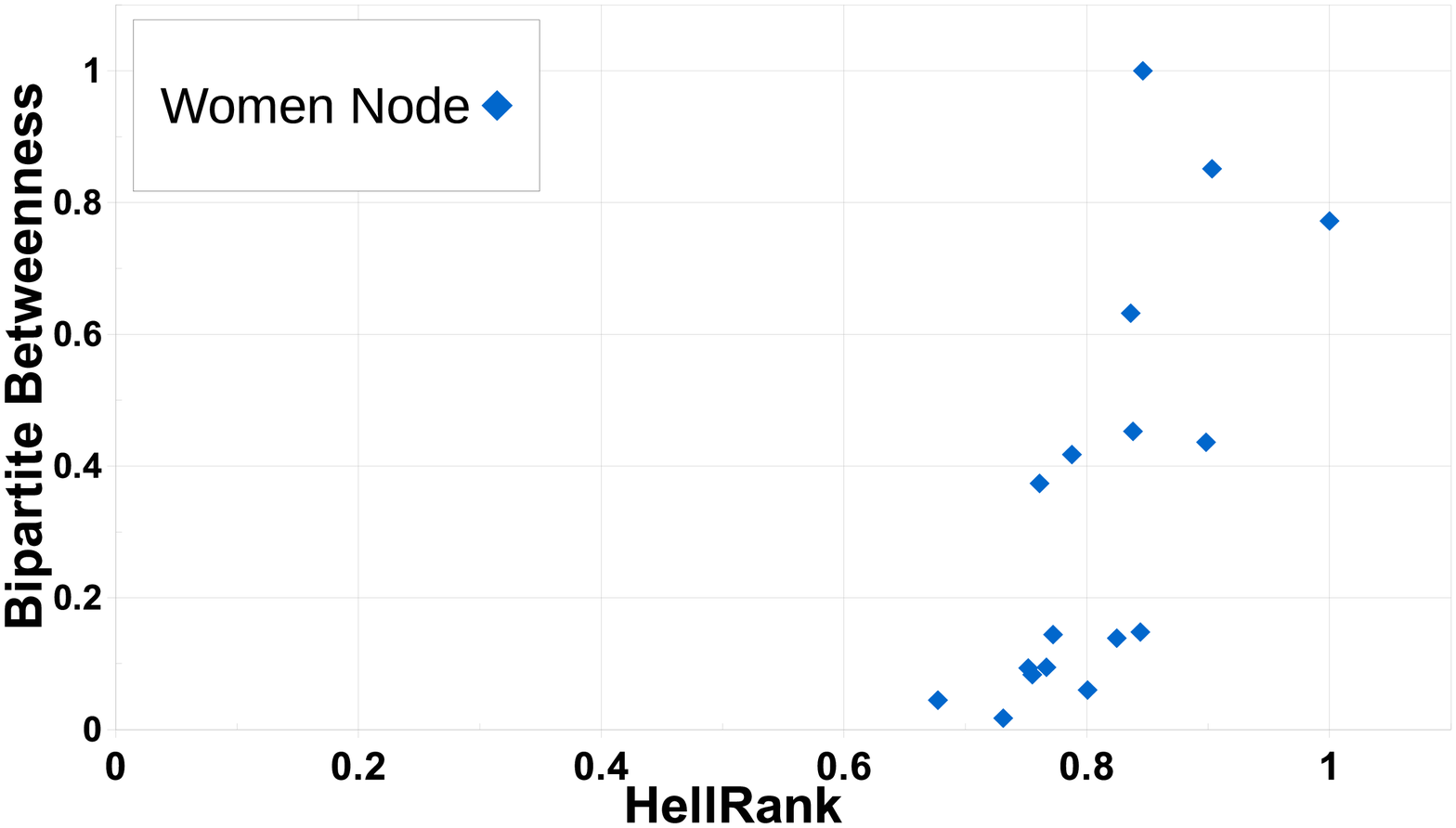}\label{davis_vs_bipartiteBetweenness}}
\hfil
\subfloat[\small HellRank vs. bipartite Degree]{\includegraphics[width=2.1in]{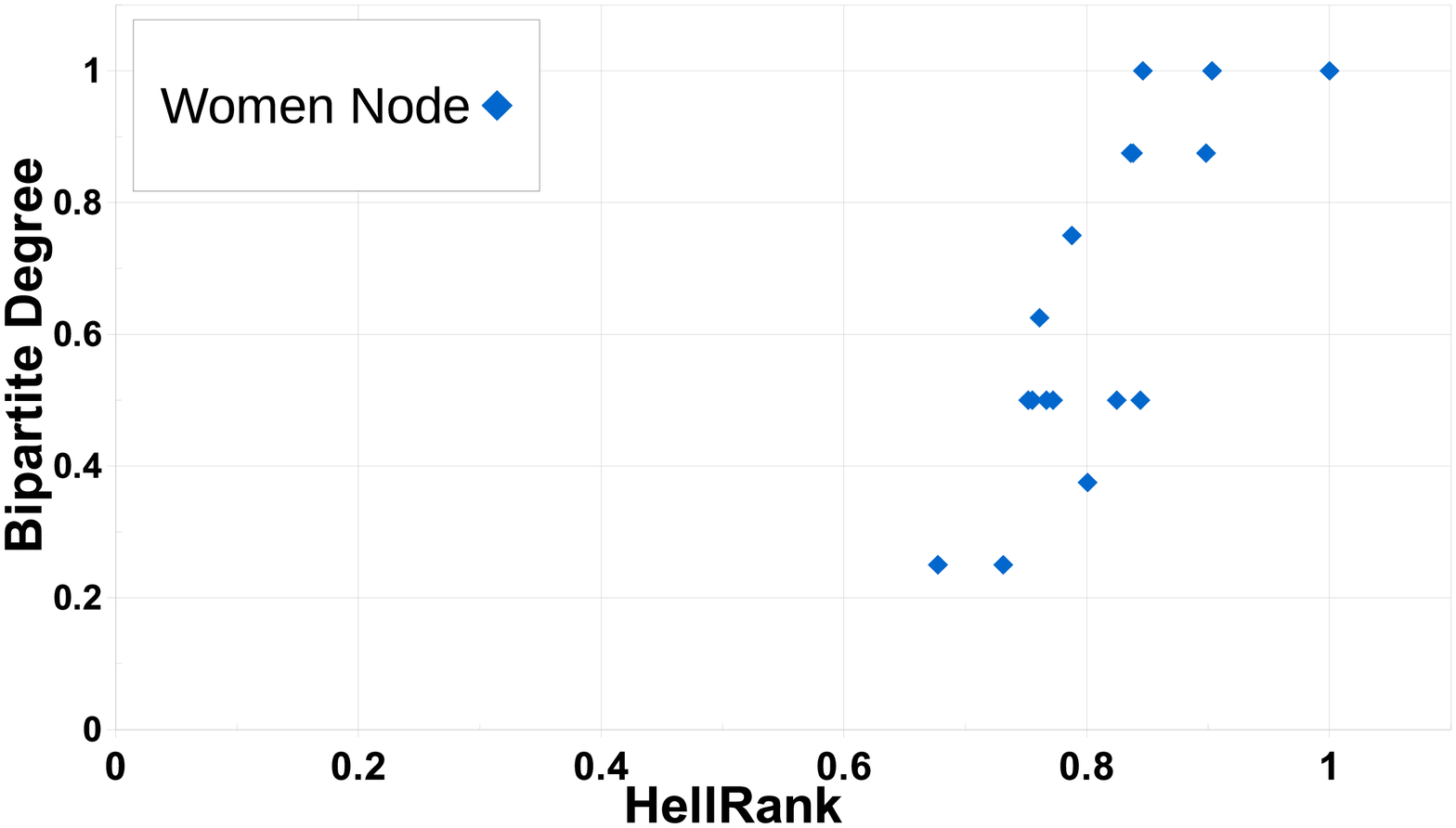}
\label{davis_vs_bipartiteDegree}}
\hfil
\subfloat[\small HellRank vs. bipartite Closeness]{\includegraphics[width=2.1in]{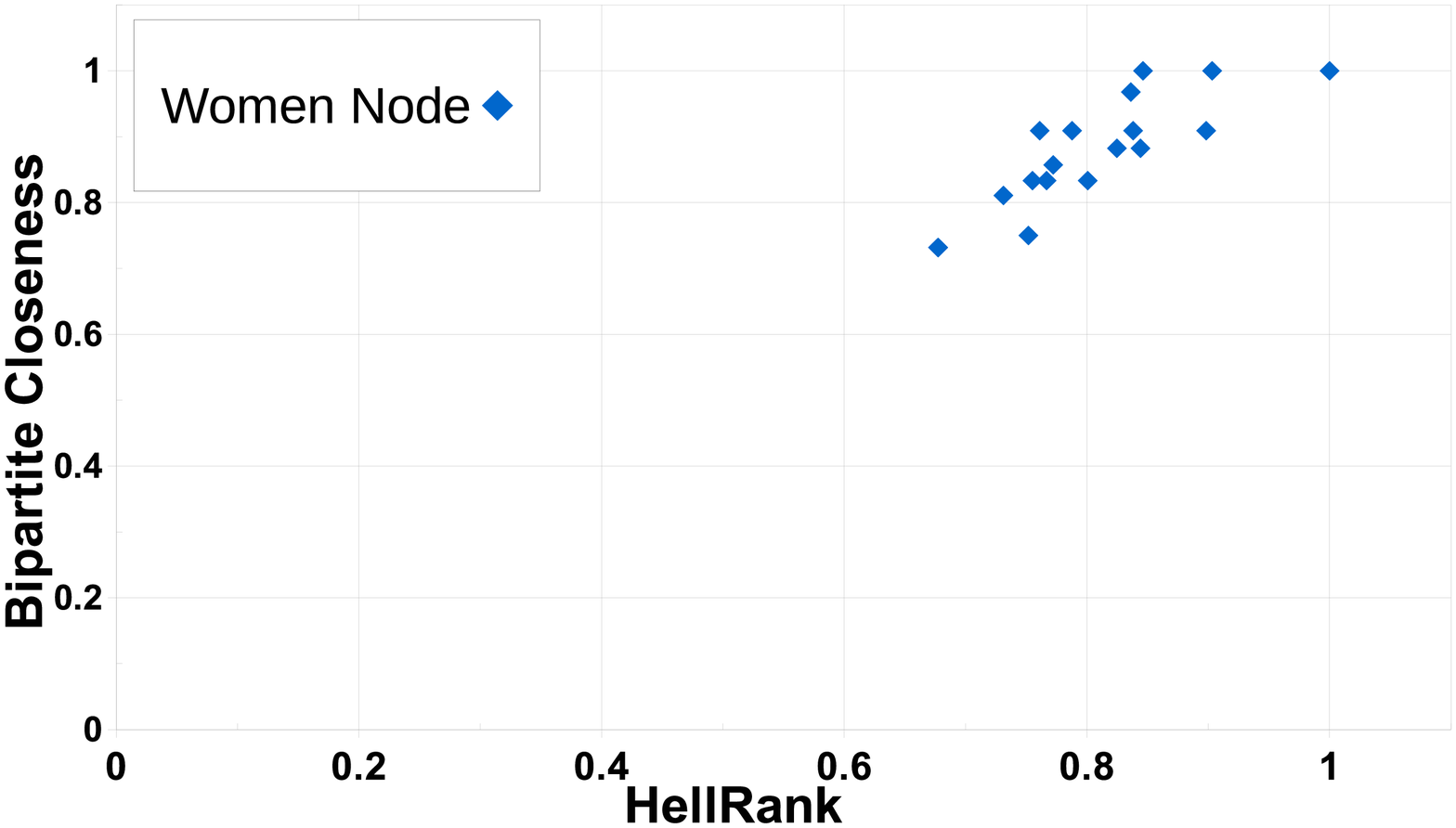}
\label{davis_vs_bipartiteCloseness}}
\hfil
\subfloat[\small HellRank vs. Latapy CC]{\includegraphics[width=2.1in]{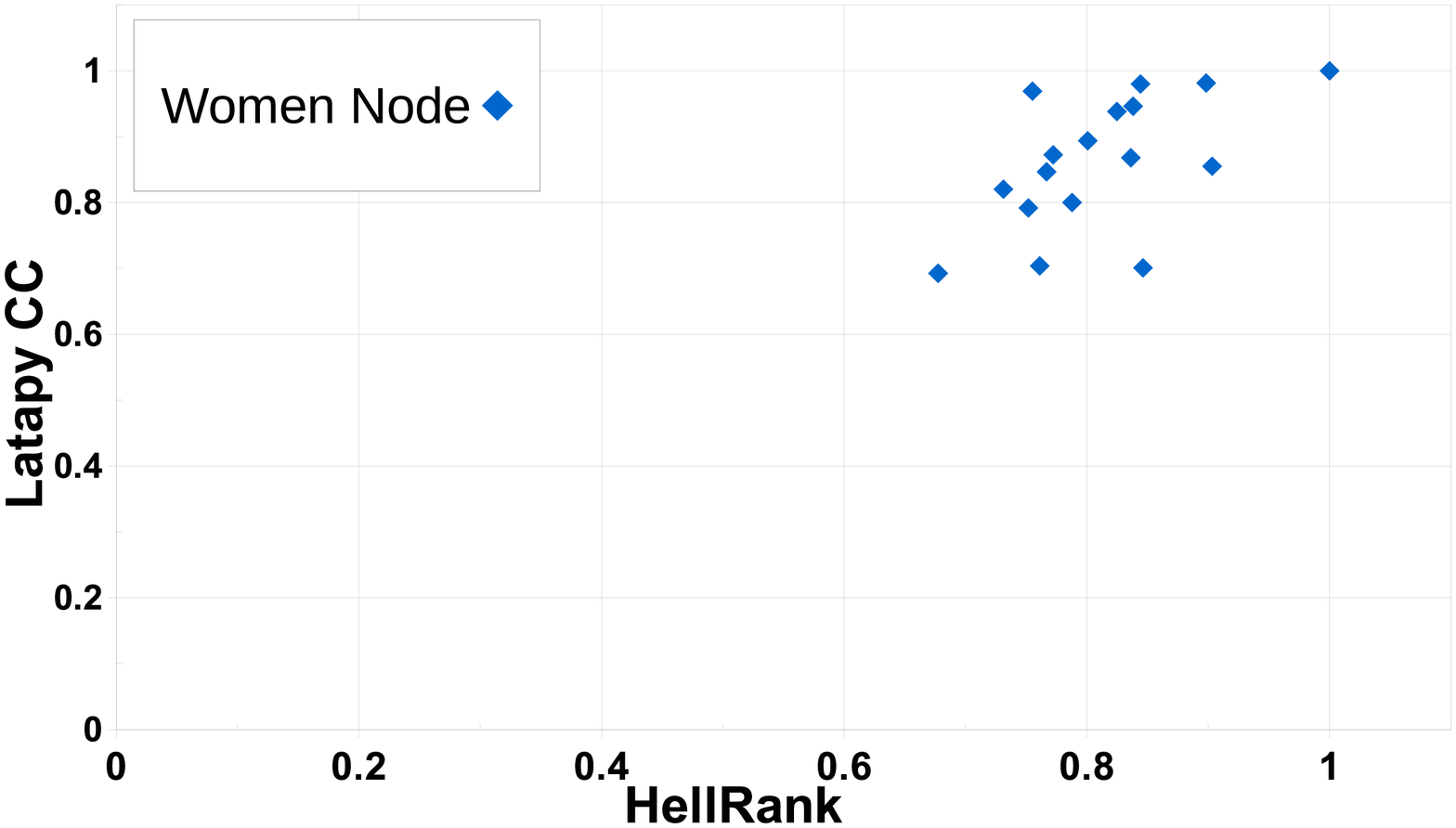}
\label{davis_vs_LatapyCC}}
\hfil
\subfloat[\small HellRank vs. PageRank]{\includegraphics[width=2.1in]{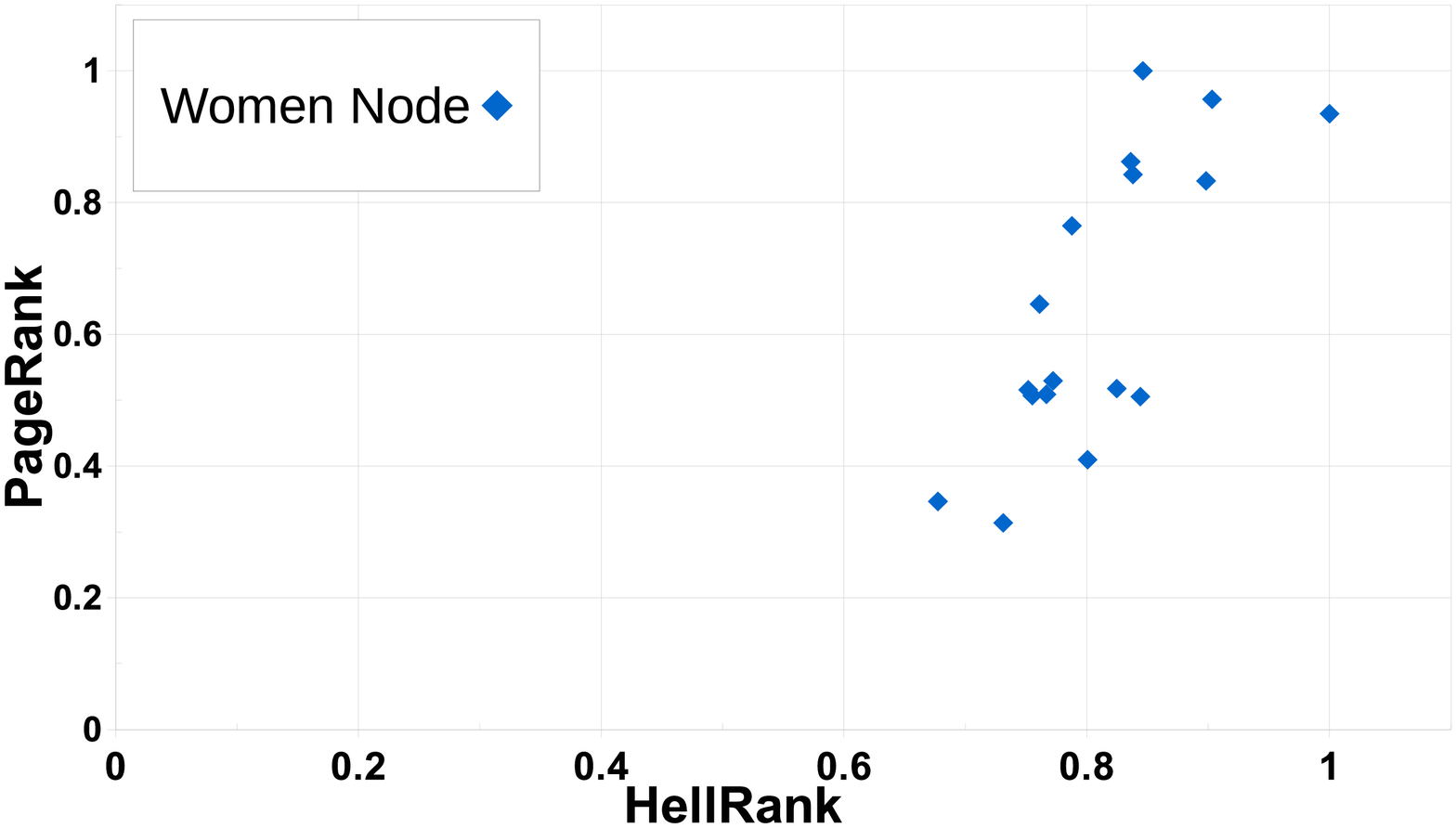}
\label{davis_vs_PageRank}}
\hfil
\subfloat[\small HellRank vs. projected Degree]{\includegraphics[width=2.1in]{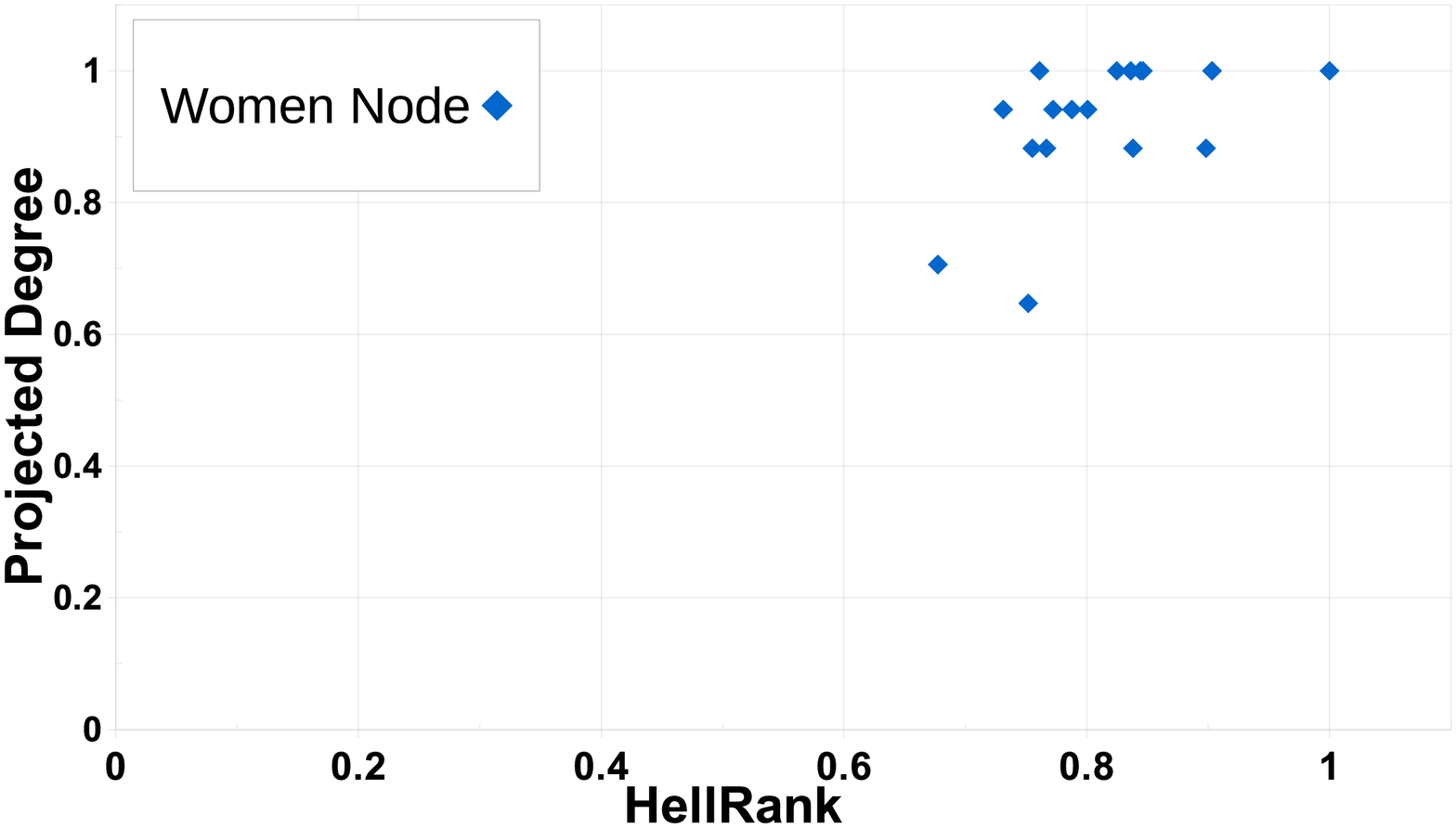}
\label{davis_vs_projected_Degree}}
\caption{\small The correlations between HellRank and the other standard centrality metrics on Davis (normalized by max value)}
\label{subfig_davis_cor}
\end{figure*}

\subsection{Correlation between HellRank and The Common Measures}
We implement our proposed HellRank measure using available tools in NetworkX \cite{Hagberg2008}. To compare our measure to the other common centrality measures such as Latapy Clustering Coefficient [Section \ref{CC}], Bipartite Degree [Section \ref{Degree}], Bipartite Closeness [Section \ref{Closeness}], Bipartite Betweenness [Section \ref{Betweenness}], and PageRank [Section \ref{PageRank}], we perform the tests on Southern Davis Women dataset. In Figure \ref{fig:compares}, we observe the obtained ratings of these metrics (normalized by maximum value) for 18 women in the Davis dataset. In general, approximate correlation can be seen between the proposed HellRank metric and the other conventional metrics in the women scores ranking. It shows that despite different objectives to identify the central users, there is a partial correlation between HellRank and the other metrics.
 
Figure \ref{subfig_davis_cor} shows scatter plots of standard metrics versus our proposed metric, on the Davis bipartite network. Each point in the scatter plot corresponds to a women node in the network. Across all former metrics, there exist clear linear correlations between each two measures. More importantly, because of the possibility of distributed computation of HellRank over the nodes, this metric can also be used in billion-scale graphs, while many of the most common metrics such as Closeness or Betweenness are limited to small networks \cite{Wehmuth2013}. We observe that high HellRank nodes have high bipartite Betweenness, bipartite Degree, and bipartite Closeness. This reflects that high HellRank nodes have higher chance to reach all nodes within short number of steps, due to its larger number of connections. In contrast with high HellRank nodes, low HellRank nodes have various Latapy CC and projected Degree values. This implies that the nodes which are hard to be differentiated by these measures can be easily separated by HellRank. 

To have a more analysis of the correlations between measures, we use Kendall between ranking scores provided by different methods in Table \ref{table:kendall_davis} and Spearman's rank correlation coefficient between top $k=5$ nodes in Table \ref{table:spearman_davis}. These tables illustrate the correlation between each pair in bipartite centrality measures and again emphasizes this point that despite different objectives to identify the central users, there is a partial correlation between HellRank and other common metrics.
\begin{table*}[t]
 \centering
\renewcommand{\arraystretch}{1.3}
\caption{\small Comparison ratings results based on Kendall score in women nodes of Davis dataset [(2) means bipartite measures and (1) means projected one-mode measures]}
\label{table:kendall_davis}
\centering
\begin{small}
\begin{tabular}{|c|c|c|c|c|c|c|c|c|}
\hline
Method &Latapy CC & Degree(2) &Betweenness(2)&Closeness(2) & PageRank &Degree(1) &Betweenness(1)&Closeness(1)\\
\hline
HellRank &0.51 &0.7  &0.67&0.74&0.59 &0.51&0.53&0.51\\
\hline
\end{tabular}
\end{small}
\end{table*}
 \begin{table*}[!t]
\renewcommand{\arraystretch}{1.3}
\caption{\small Comparison top $k=5$ important nodes based on  Spearman's correlation in women nodes of Davis dataset [(2) means bipartite and (1) means projected one-mode]}
\label{table:spearman_davis}
\centering
\begin{small}
\begin{tabular}{|c|c|c|c|c|c|c|c|c|}
\hline
Method &Latapy CC & Degree(2) &Betweenness(2)&Closeness(2)& PageRank &Degree(1) &Betweenness(1)&Closeness(1)\\
\hline
HellRank &0.44 &0.72&0.44&0.72&0.44 &0.44&0.44&0.44\\
\hline
\end{tabular}
\end{small}
%\vspace{10pt}
\end{table*}
\begin{figure*}[!t]
\centering
\includegraphics[width=6in]{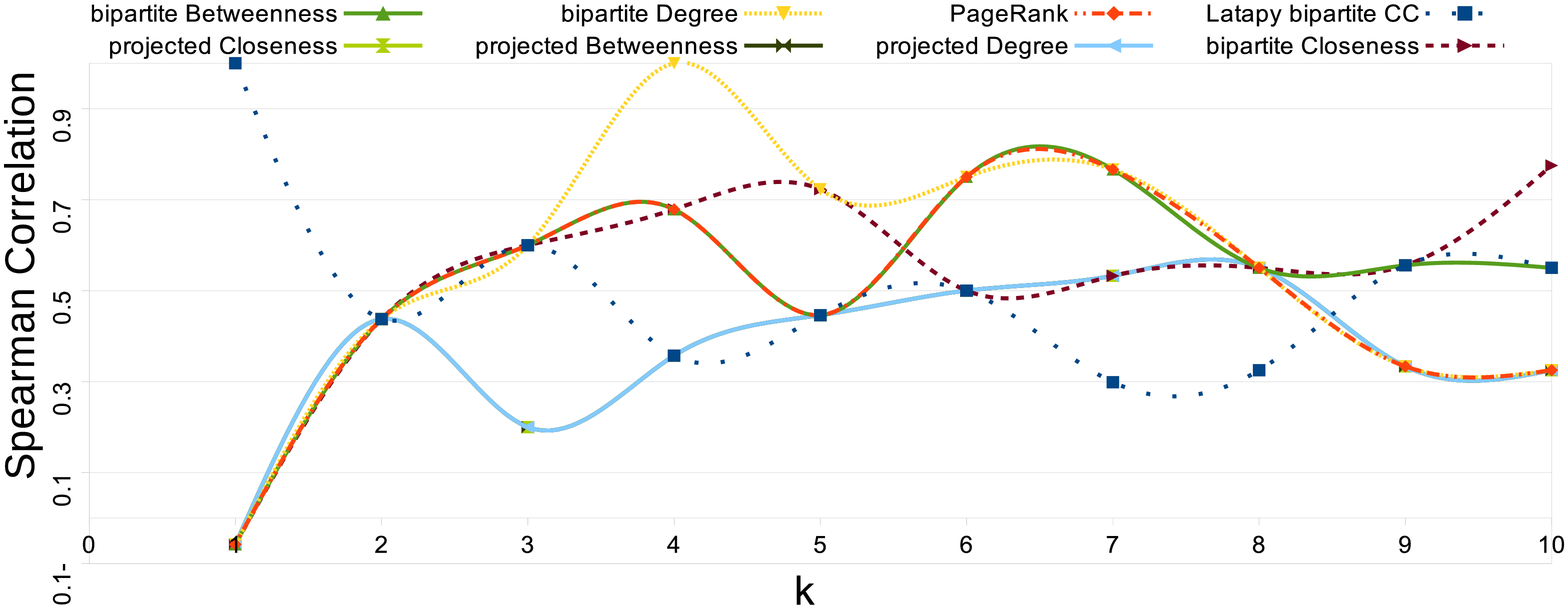}
\vspace{-10pt}
\caption{\small Spearman's rank correlation with different values of $k$ in Davis dataset}
%\vspace{-3pt}
\label{davis_corr_dif_K}
\end{figure*}

In the next experiment, we compare the top-$k$ central users rankings produced by Latapy CC, PageRank, Bipartite, and projected one-mode Betweenness, Degree, Closeness, and HellRank with different values of $k$. We employ Spearman's rank correlation coefficient measurement to compute the ranking similarity between two top-$k$ rankings. Figure \ref{davis_corr_dif_K} presents result of Spearman's rank correlation coefficient between the top-$k$ rankings of HellRank and the other seven metrics, in terms of different values of $k$. As shown in the figure,  the correlation values of top $k$ nodes in all rankings,  reach almost constant limit at a specific value of $k$. This certain amount of $k$ is approximately equal 4 for all metrics. This means that the correlation does not increase at a certain threshold for $k=4$ in the Davis dataset.

To evaluate the HellRank on a larger dataset, we repeated all the mentioned experiments for the arXiv cond-mat dataset. The scatter plots of standard metrics versus HellRank metric can be seen in Figure \ref{subfig_arxiv_cor}. The results show that there exist almost linear correlations between the two measures in Bipartite Betweenness, Bipartite Degree and PageRank. In contrast to these metrics, HellRank has not correlation with other metrics such as Bipartite Closeness, Latapy Clustering Coefficient and Projected Degree. This implies that nodes that are hard to be differentiated by these metrics, can be separated easily by HellRank metric. 

Moreover, Spearman's rank correlation with different values of $k$ in arXiv cond-mat dataset can be seen in Figure \ref{opsahl_corr_dif_K}. We observe the correlation values from top $k$ nodes in all rankings, with different values of $k$,  reach almost constant limit at a specific value of $k$. This certain amount of $k$ approximately equals to 1000 for all metrics except Bipartite Closenss and Latapy CC metrics. This means that the correlation does not increase at a certain threshold for $k=1000$ in the arXiv cond-mat dataset.

\begin{figure*}[!t]
\centering
\subfloat[\small HellRank vs. bipartite Betweenness]{\includegraphics[width=2.1in]{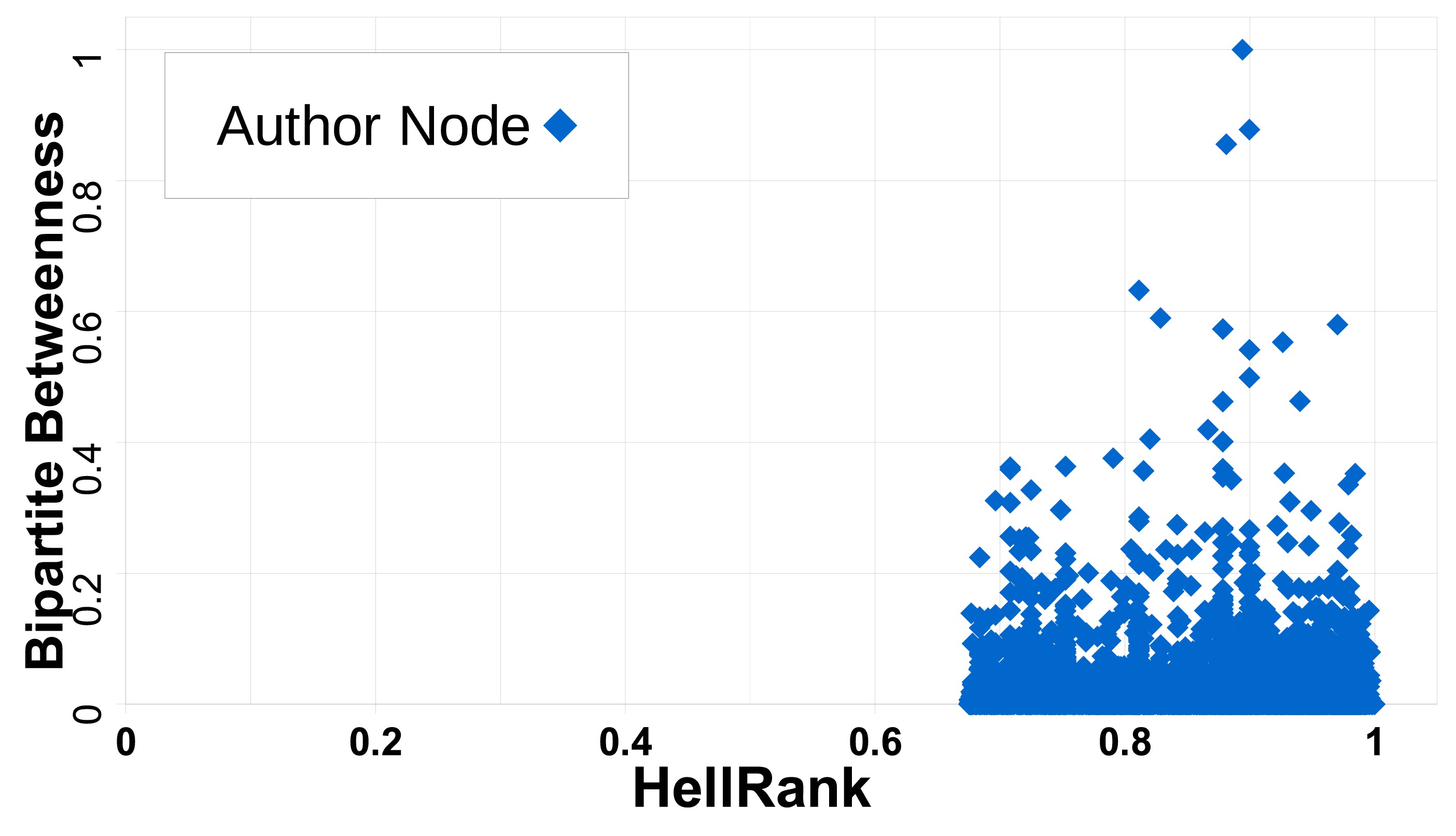}\label{opsahl_vs_bipartiteBetweenness}}
\hfil
\subfloat[\small HellRank vs. bipartite Degree]{\includegraphics[width=2.1in]{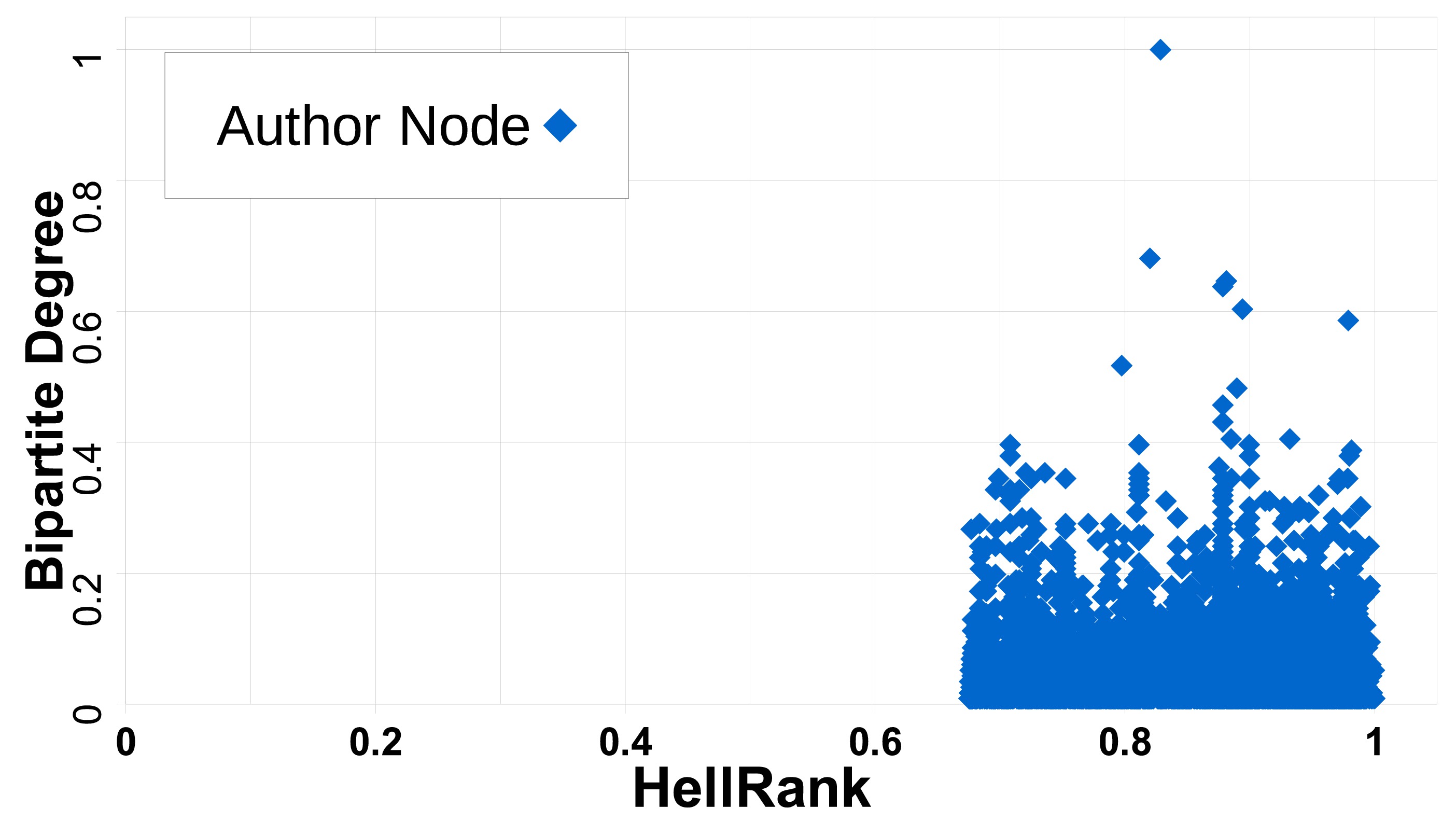}
\label{opsahl_vs_bipartiteCloseness}}
\hfil
\subfloat[\small HellRank vs. bipartite Closeness]{\includegraphics[width=2.1in]{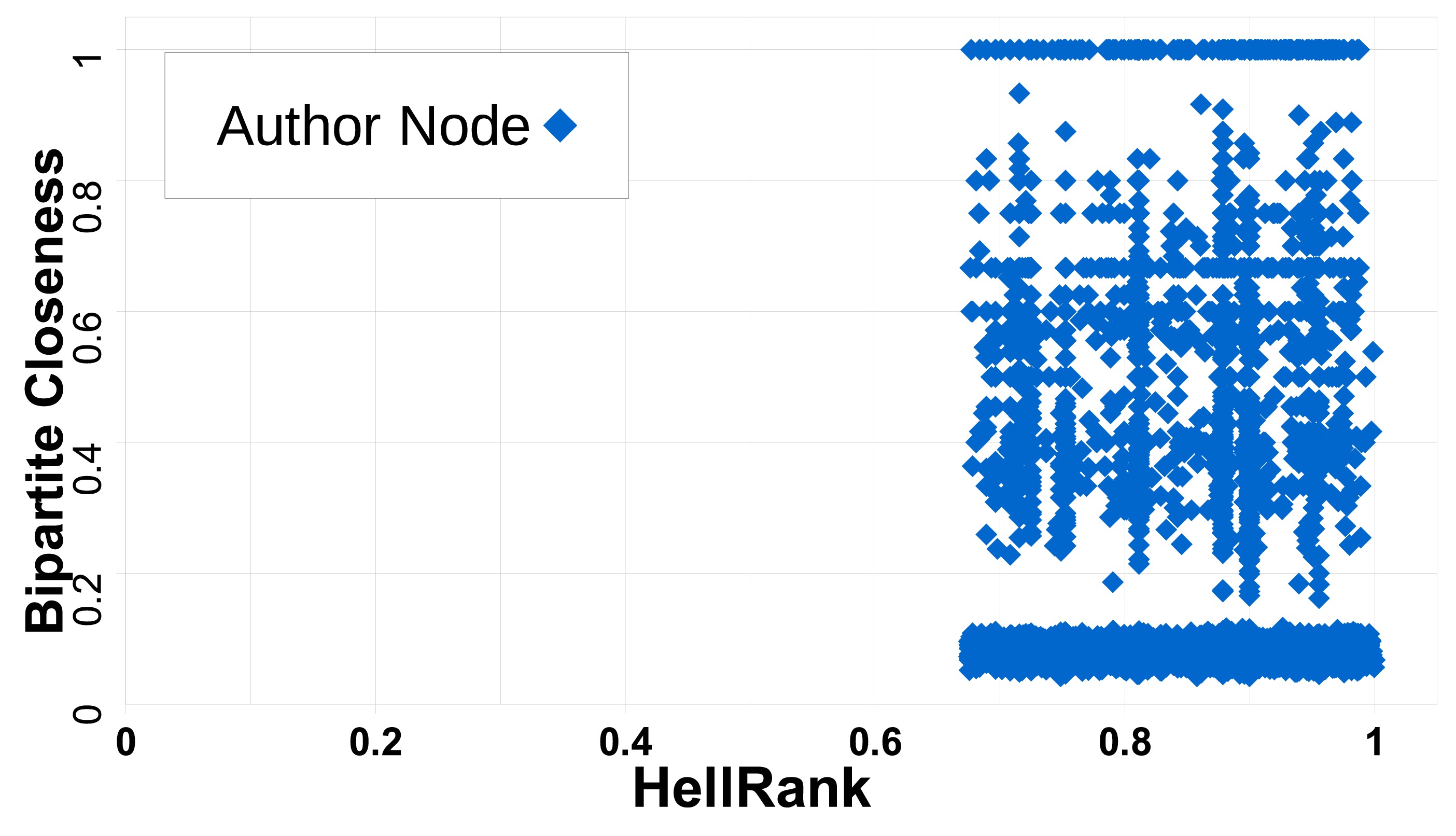}
\label{opsahl_vs_bipartiteCloseness}}
\hfil
\subfloat[\small HellRank vs. Latapy CC]{\includegraphics[width=2.1in]{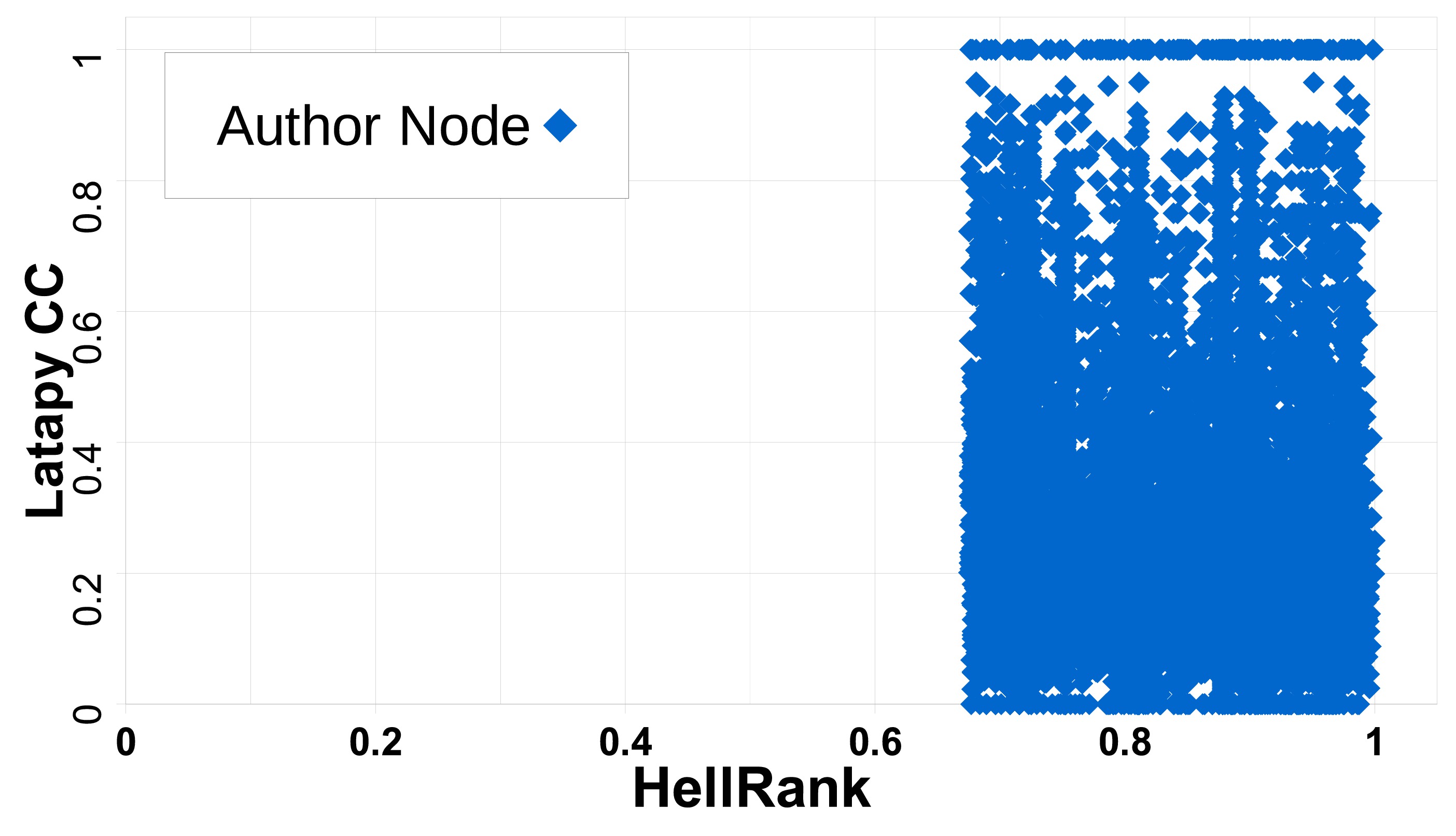}
\label{fopsahl_vs_LatapyCC}}
\hfil
\subfloat[\small HellRank vs. PageRank]{\includegraphics[width=2.1in]{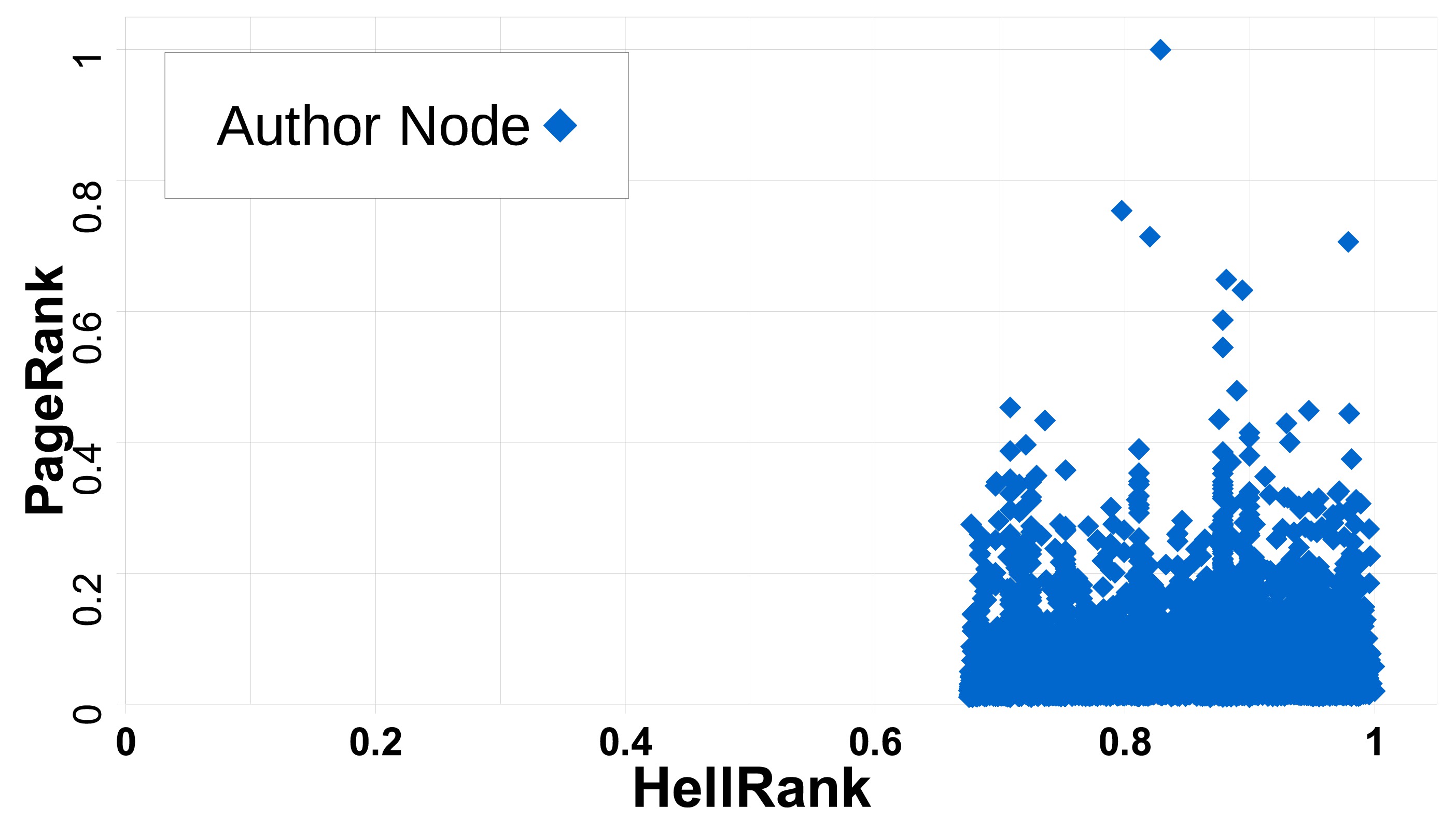}
\label{opsahl_vs_PageRank}}
\hfil
\subfloat[\small HellRank vs. projected Degree]{\includegraphics[width=2.1in]{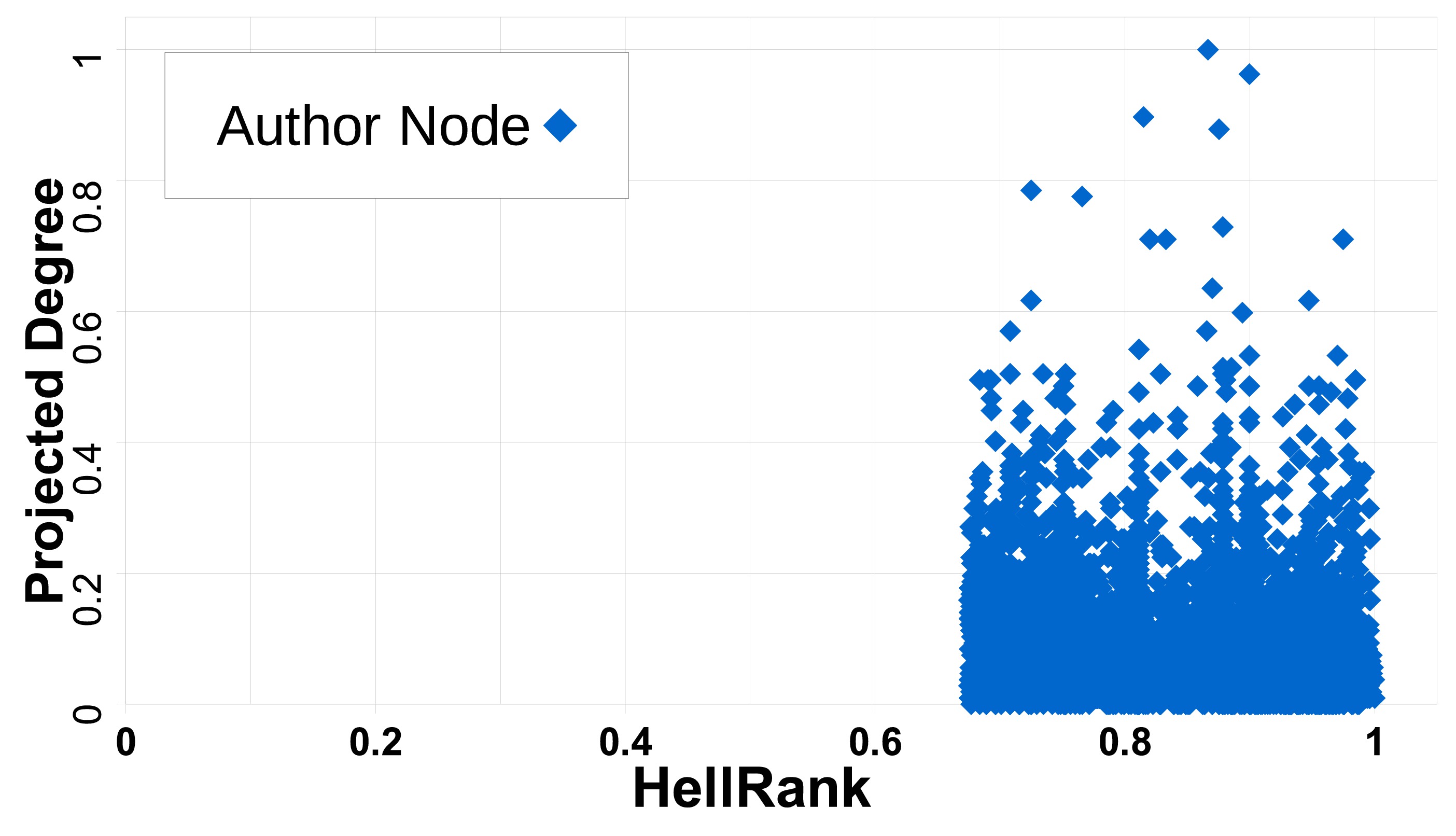}
\label{opsahl_vs_projected_Degree}}
\caption{\small The correlations between HellRank and the other standard centrality metrics on Opsahl (normalized by max value)}
\vspace{7pt}
\label{subfig_arxiv_cor}
\end{figure*}

\begin{figure*}[t]
\centering
\includegraphics[width=6in]{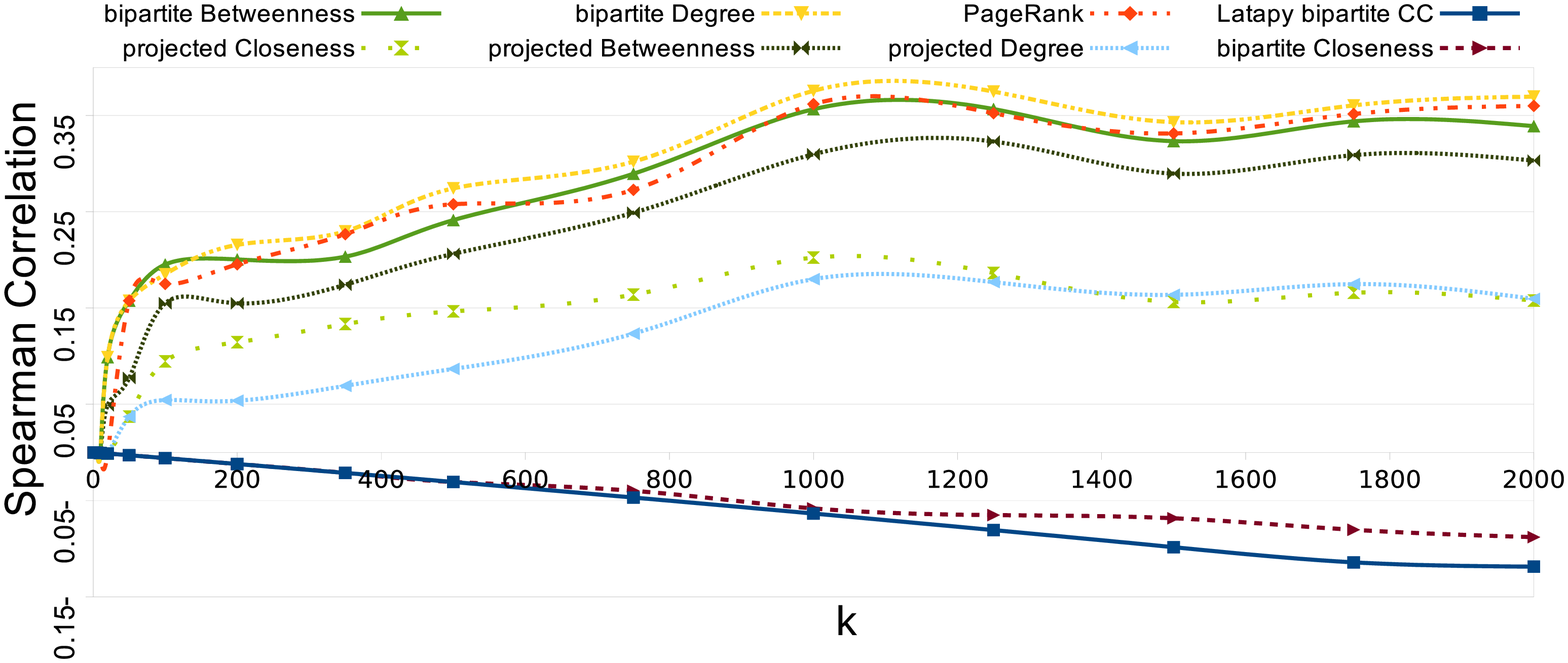}
\vspace{-10pt}
\caption{\small Spearman's rank correlation with different values of $k$ in arXiv cond-mat dataset}
%\vspace{-3pt}
\label{opsahl_corr_dif_K}
\end{figure*}

\subsection{Mapping the Hellinger Distance Matrix to the Euclidean Space}
Since we have a well-defined metric features and ability of mapping the Hellinger distance matrix to the Euclidean space, other experiment that can be done on this matrix,  is clustering nodes based on their distance. This Hellinger distance matrix can then be treated as a valued adjacency matrix\footnote{In a valued adjacency matrix, the cell entries can be any non-negative integer, indicating the strength or number of relations of a particular type or types \cite{koput2010}.} and visualized using standard graph layout algorithms. Figure \ref{map_euclidean} shows the result of such an analysis on Davis dataset. This figure is a depiction of Hellinger Distance for each pair of individuals, such that a line connecting two individuals indicates that their Hellinger distance are less than 0.50. The diagram clearly shows the separation of Flora and Olivia, and the bridging position of Nora. 
\begin{figure*}[t]
\centering
\includegraphics[width=5.2in]{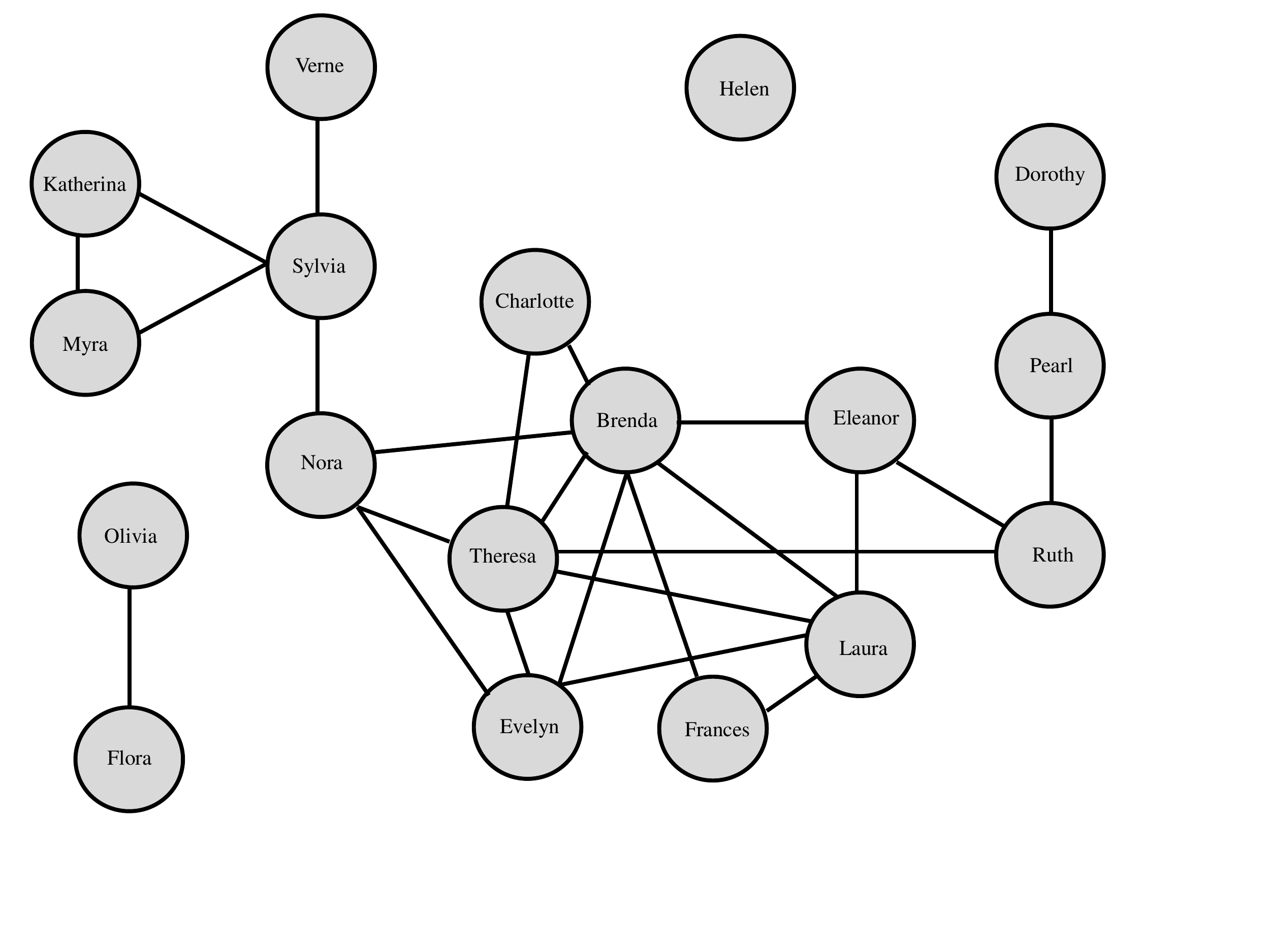}
\vspace{-25pt}
\caption{\small Mapping Hellinger distance matrix to Euclidean Space. A tie indicates that the distance between two nodes is lesser than 0.50}
\label{map_euclidean}
\end{figure*}

\section{Conclusion and Future Work}

In this paper, we proposed HellRank centrality measure for properly detection of more behavioral representative users in bipartite social networks. As opposed to previous work, by using this metric we can avoid projection of bipartite networks into one-mode ones, which makes it possible to take much richer information from the two-mode networks. The computation of HellRank can be distributed by letting each node uses only local information on its immediate neighbors. To improve the accuracy of HellRank, we can extend the neighborhood around each node. The HellRank centrality measure is based on the Hellinge distance between two nodes of the bipartite network and we theoretically find the upper and the lower bounds for this distance.

We experimentally evaluated HellRank on the Southern Women Davis dataset and the results showed that Brenda, Evelyn, Nora, Ruth, and Theresa should be considered as important women. Our evaluation analyses depicted that the importance of a woman does not only depend on her Degree, Betweenness, and Closeness centralities. For instance, if Brenda with low Degree centrality is removed from the network, the information would not easily spread among other women. As another observation, Dorothy, Olivia, and Flora have very low HellRank centralities. These results are consistent with the results presented in Bonacich (1978), Doreian (1979), and Everett and Borgatti (1993). 

As a future work, more meta data information can be taken into account besides the links in a bipartite network. Moreover, we can consider a bipartite network as a wighted graph \cite{mahyar2013ucswn,Mahyar2015LSRweighted} in which the links are not merely binary entities, either present or not, but have associated a given wight that record their strength relative to one another. Furthermore, as HellRank measure is proper for detection of more behavioral representative users in bipartite social network, we can use this measure in Recommender Systems. In addition, we can detect top $k$ central nodes in a network with indirect measurements and without full knowledge the network topological structure, using compressive sensing theory \cite{Mahyar2015TopK,mahyar2013ucsnt,mahyar2013ucswn,Mahyar2015CScomdet,Mahyar2015LSRweighted}.

%\begin{acknowledgements}
%\end{acknowledgements}

\bibliographystyle{spmpsci}      % mathematics and physical sciences
\bibliography{ref}

%\begin{thebibliography}{}
% etc
%\end{thebibliography}

\end{document}